\documentclass[thmsa,11pt,leqno,11pt]{article}
\usepackage{amssymb}
\usepackage{amsfonts}
\usepackage{graphicx}
\usepackage[leqno]{amsmath}
\usepackage{scalefnt}
\usepackage{fancybox}
\usepackage{float}
\usepackage{natbib}
\usepackage{theorem}
\usepackage{lscape}

\newtheorem{theorem}{Theorem}

\newtheorem{assumption}{Assumption}

\newtheorem{condition}{Condition}

\newtheorem{corollary}[theorem]{Corollary}

\newtheorem{definition}{Definition}

\newtheorem{lemma}{Lemma}

\newtheorem{proposition}[theorem]{Proposition}
{\theorembodyfont{\upshape} 

}

\newenvironment{proof}[1][Proof]{\textbf{#1.} }{\ \rule{0.5em}{0.5em}}
\numberwithin{theorem}{section}
\numberwithin{assumption}{section}
\numberwithin{definition}{section}
\numberwithin{example}{section}
\numberwithin{lemma}{section}
\numberwithin{remark}{section}
\numberwithin{algorithm}{section}
\numberwithin{equation}{section}

\renewcommand{\cite}{\citet}
\numberwithin{Assumption}{section}

\begin{document}

\title{Quantile Regression with Interval Data}
\author{Arie Beresteanu\thanks{%
Department of Economics, University of Pittsburgh, arie@pitt.edu. } \and
Yuya Sasaki\thanks{%
Department of Economics, Vanderbilt University, yuya.sasaki@vanderbilt.edu.}}
\date{\today\thanks{%
We have benefited from discussions with Tong Li, Francesca Molinari, Ilya Molchanov, Tatsushi Oka, Roee Teper, and Adam Rosen and very useful comments by Esfandiar Maasoumi (the editor), the editors of this special issue in honor of Cheng Hsiao, an anonymous referee, and participants in NASM 2016, New York Camp Econometrics XIII, CMES 2018, and IAAE 2018. All remaining errors are ours. This research was sponsored
in part by NSF grant SES-0922373.}}
\maketitle

\begin{abstract}\setlength{\baselineskip}{5.5mm}
This paper investigates the identification of quantiles and quantile regression parameters when observations are set valued.
We define the identification set of quantiles of random sets in a way that extends the definition of quantiles for regular random variables.
We then give sharp characterization of this set by extending concepts from random set theory.
Applying the identification set of quantiles and its sharpness to parametric quantile regression models yields the identification set of the parameters and its sharpness.
We apply our methods to data on localized environmental benefits and their impact on house values.
\bigskip

\noindent \textbf{Keywords}: Partial Identification, Random Sets, Quantile Regression, Quantile Sets.

\noindent \textbf{JEL Code}: C21.

\bigskip
\end{abstract}

\newpage

\section{Introduction}

Interval valued observations are common in data based on surveys. 
One type of interval valued data is generated by a response to a questioner that offers a distinct set of intervals to choose from. 
Another type of interval data is generated by respondents choosing the minimum and/or maximum amount.
Willingness to pay surveys often fall in this type. 
Simple econometric tools, such as to apply the OLS using the midpoint of willingness-to-pay interval as a dependent variable, have long been known to suffer from substantial biases -- see \cite{CameronHuppert88}.

Set identification and set inference approaches are proposed to solve this issue under a number of econometric contexts -- see our discussion of the literature ahead near the end of this section.
To the best of our knowledge, however, no preceding paper in this literature has discussed quantile regression models with general quantile ranks and general set valued data.
Some empirical papers \citep[e.g.,][]{OGarraMourato2007,GamperTimmins13}, on the other hand, use quantiles and quantile regressions where outcome data are interval-valued by taking the midpoint of the interval as the representative value.
In this light, this paper investigates the identification of quantiles and of quantile regressions when the outcome variable is set valued.

We first identify the unconditional and conditional quantiles of a random set in Section \ref{sec:identification}. 
The concept of the quantile set of a random set is introduced. 
The quantile set is shown to be identified by the containment and capacity quantiles which we define in Section \ref{sec:identification}.
The identification argument for unconditional quantiles extends to that for conditional quantiles in Section \ref{sec:nonparametric_conditional_quantile}.
Further, the identification argument for conditional quantiles is extended to set identification of quantile regression functions depending on a finite number of parameters in Section \ref{sec:quantile_regression}.
We show that this identification set is defined by a system of conditional moment inequalities.
These inequalities involve the cumulative containment and cumulative capacity functionals, which characterize the identification set for the quantile regression parameters.
We also show that the sharp identification set is convex if the quantile regression is linear in parameters.
We use simulation studies to demonstrate the validity of the sharp characterization of the identified sets in Section \ref{sec:simulation_studies}.


\textbf{Literature:}
This paper is related to two broad literatures.
One is the literature on partial identification with interval data and the other is the literature on quantile regressions.
For the former branch of the literature, binary choice models with interval regressors are discussed in \cite{ManskiTamer02}. 
Mean regressions when the outcome variables are interval valued are discussed in \cite{BeresteanuMolinari08}. 
The models of \cite{ManskiTamer02} and \cite{BeresteanuMolinari08} are generalized in \cite{BeresteanuMolchanovMolinari11}. 

Quantile regressions are introduced and studied extensively in the literature \citep[e.g.,][]{KoenkerBasset78,Koenker05}.
A number of notable papers in the literature discuss identification of quantiles and/or quantile regressions under set-valued observed outcomes.
One of the most common causes of set-valued outcomes is censoring.
\citet{Powell1984} provides an estimator for the linear median regression model where the outcome variable is censored.
\citet{Manski1985} discusses identification of the linear median regression model where econometricians only observe the sign of an outcome variable.
\citet{HongTamer2003} discuss inference on the linear median regression model where the outcome variable is censored and regressors are endogenous.
Khan and Tamer (2009) discuss inference on the linear median regression model where the outcome variable is endogenously censored.
The model considered in this paper includes the case of censored outcome variables as special cases of set-valued outcome variables. 
Furthermore, compared with these preceding papers, we consider generalized quantile ranks $\tau \in (0,1)$ in addition to the median $\tau = 0.5$.

More recently, \cite{LiOka2015} consider linear quantile regressions with a censored outcome variable in the framework of \cite{Rosen2012}.
We do not deal with panel data, but the model considered in this paper includes the case of censored outcome variables as argued above.
While the source of partial identification is not interval-valued outcomes, partial identification of nonseparable models are also investigated by \citet{Chesher2005,Chesher2010} and generalized by \citet{ChesherRosen2015}.


\textbf{Notations and definitions:}
We introduce basic notations and definitions partly following those of \cite{Molchanov05}.
Let $\left( \Omega ,\Im ,\mathbf{P}\right) $ be a complete probability space on which all random variables and sets are defined. 
Let $\mathcal{K}(\mathbb{R})$ denote the collection of all closed sets in $\mathbb{R}$.
For an $\mathbb{R}$-valued random variable $y$, let $F_{y}\left( t\right) =\mathbf{P}\left(\{\omega: y(\omega) \in \left( -\infty ,t\right] \}\right) $ define the cumulative distribution function $F_y$ of $y$. 
When it exists, the probability density function is denoted by $f_y$.
For $\tau \in \left( 0,1\right) $, let $q_{y}\left( \tau \right) = \inf \left\{ t: F_{y}\left( t\right) \geq \tau \right\} $ denote the $\tau$-th quantile of $y$. 
A random variable $y: \Omega \rightarrow \mathbb{R}$ is a measurable selection of $Y: \Omega \rightarrow \mathcal{K}(\mathbb{R})$ if $y(\omega) \in Y(\omega)$ $\mathbf{P}$-a.s. 
The set of selections of $Y$ is denoted by $Sel(Y)$.
The containment functional $C_Y$ and capacity functional $T_Y$ of $Y$ are defined by $C_Y(K) = \mathbf{P}(\left\{ \omega : Y(\omega) \subset K\right\})$ and $T_Y(K) = \mathbf{P}(\left\{ \omega : Y(\omega) \cap K \neq \emptyset\right\})$, respectively.


\section{Quantiles of Random Sets}\label{sec:identification}

We start by discussing identification of quantiles. 
In the current section, we focus on unconditional quantiles.
Section \ref{sec:covariates} extends this baseline result to conditional quantiles by considering the corresponding conditional probabilities given covariates, where the responses are allowed to arbitrarily depend on the covariates.

\begin{assumption}[Data]\label{a:unconditional}
Let $(y^\ast,Y): \Omega \rightarrow \mathbb{R} \times \mathcal{K}(\mathbb{R})$ be such that 
$y^{\ast }$ is unobserved,
$Y$ is observed,
$Y(\omega)$ is non-empty $\mathbf{P}$-a.s., and
$y^{\ast }(\omega) \in Y(\omega),$ $\mathbf{P}$-a.s.
\end{assumption}

At this point, we do \textit{not} assume that $Y(\omega)$ is interval-valued.
We later show that sharp identification requires $Y(\omega)$ to be interval-valued -- see Theorem \ref{prop:id_convex} ahead.

We would like to learn about $q_{y^{\ast }}\left( \tau \right) $. 
Define the $\tau$-th quantile set of $Y$ by
\begin{equation*}
\Theta^Y_0\left( \tau \right) =\left\{ q_{y}\left( \tau \right) :y\in Sel\left( Y\right) \right\} .
\end{equation*}
This is, by definition, the identification set for $q_{y^{\ast }}\left( \tau \right) $.
In other words, with no further information the only thing we can say about $q_{y^{\ast }}\left( \tau \right) $ is that $q_{y^{\ast }}\left( \tau \right) \in \Theta^Y_0\left( \tau \right) $.

For any $t\in \mathbb{R}$, define 
\begin{eqnarray*}
\tilde{C}_{Y}\left( t\right) &=& C_{Y}\left( (-\infty ,t]\right) \ \text{and} \\
\tilde{T}_{Y}\left( t\right) &=& T_{Y}\left( (-\infty ,t]\right)
\end{eqnarray*}
to be
the cumulative containment and cumulative capacity functionals, respectively. 
Note that $\tilde{C}_{Y}$ and $\tilde{T}_{Y}$ are monotone increasing and right continuous. 
For $\tau \in \left( 0,1\right) $, define 
\begin{eqnarray*}
\tilde{C}_{Y}^{-1}\left( \tau \right) &=&\inf \left\{ t:\tilde{C}_{Y}\left(t\right) \geq \tau \right\} \ \text{and} \\
\tilde{T}_{Y}^{-1}\left( \tau \right) &=&\inf \left\{ t:\tilde{T}_{Y}\left(t\right) \geq \tau \right\}
\end{eqnarray*}
to be the containment and capacity quantiles of $Y$, respectively.
Since $Y$ is $\mathcal{K}(\mathbb{R})$-valued, $\tilde{C}_{Y}^{-1}$ and $\tilde{T}_{Y}^{-1}$ are equivalent to $q_{\sup Y}$ and $q_{\inf Y}$, respectively.

\begin{theorem}[Partial Identification]\label{prop:id}
Suppose that Assumption \ref{a:unconditional} holds.
For every $\tau \in \left( 0,1\right) $, $\Theta^Y_0\left( \tau \right) \subset \left[ q_{\inf Y}\left( \tau \right) ,q_{\sup Y}\left( \tau \right) \right]$.
\end{theorem}

\begin{proof}
The statement trivially holds if $\Theta^Y_0(\tau)$ is an empty set. 
Suppose that $\Theta^Y_0(\tau)$ is not empty.
Let $t^\ast \in \Theta^Y_0(\tau) = \left\{ q_y(\tau) : y \in Sel(Y) \right\}$.
Thus, there exists $y \in Sel(Y)$ such that $t^\ast = q_y(\tau)$.
By the definitions of $\tilde C_Y$ and $\tilde T_Y$, $y\in Sel\left( Y\right) $ implies $\tilde C_{Y}\left(t\right) \leq F_{y}\left( t\right) \leq \tilde T_{Y}\left( t \right) $ for all $t \in \mathbb{R}$.\footnote{This implication corresponds to the necessity part of Artstein's Lemma, which does not restrict to compact sets. See discussion in \cite{BeresteanuMolchanovMolinari2012}.}
Therefore, $\inf \left\{ t:\tilde{T}_{Y}\left( t\right) \geq \tau \right\} \leq \inf \left\{ t:F_{y}\left(t\right) \geq \tau \right\} \leq \inf \left\{ t:\tilde{C}_{Y}\left( t\right)\geq \tau \right\} $.
This proves the theorem.
\end{proof}

When $Y=\left\{ y \right\} $ for a random variable $y $, $\Theta^Y_0\left( \tau \right) =\left\{ q_{y }\left( \tau \right) \right\} $. 
\cite{Molchanov1990} defines a quantile of a random set in a different way.
His definition does not provide the same generalization of a quantile function in case of a $\mathbb{R}$-valued random set that we need in this paper. 

The other direction of set inclusion for Theorem \ref{prop:id} need not hold.
To see this, consider a simple random set $Y$ such that $Y(\omega) = [-2,-1] \cup [1,2]$ $\mathbf{P}$-a.s.
Then, we have $0 \in [-2,2] = \left[ q_{\inf Y}\left( 0.5 \right) ,q_{\sup Y}\left( 0.5 \right) \right]$, but $0 \not\in \Theta^Y_0(0.5) \subset [-2,-1] \cup [1,2]$.
This example illustrates why a `hole' in the set $Y(\omega)$ fails to establish the other direction of set inclusion.
This observation in fact can be generalized.
The following theorem shows that the identification set equality holds without such holes (i.e., for interval-valued random sets).

\begin{theorem}[Sharpness]\label{prop:id_convex}
Suppose that Assumption \ref{a:unconditional} holds.
If $Y$ is a convex valued random set in $\mathcal{K}(\mathbb{R})$ with $Sel(Y) \neq \emptyset$, then $\left(q_{\inf Y}\left( \tau \right) ,q_{\sup Y}\left( \tau \right) \right) \subset \Theta^Y_0\left( \tau \right)$ for all $\tau \in \left( 0,1\right)$.
Furthermore, if in addition $\inf Y(\omega) \in Y(\omega) > -\infty$ $\mathbf{P}$-a.s., then $q_{\inf Y}\left( \tau \right) \in \Theta^Y_0\left( \tau \right)$ for all $\tau \in \left( 0,1\right)$.
Similarly, if in addition $\sup Y(\omega) < \infty$  $\mathbf{P}$-a.s., then $q_{\sup Y}\left( \tau \right) \in \Theta^Y_0\left( \tau \right)$ for all $\tau \in \left( 0,1\right)$.
\end{theorem}

\begin{proof}
If $\left( \tilde{T}_{Y}^{-1}\left( \tau \right) ,\tilde{C}_{Y}^{-1}\left( \tau \right) \right)$ is empty, then the first claim in the theorem is trivially satisfied.
Now, suppose that it is non-empty.
Fix $\tau \in \left( 0,1\right) $ and take $t\in \left( \tilde{T}_{Y}^{-1}\left( \tau \right) ,\tilde{C}_{Y}^{-1}\left( \tau \right) \right)$. 
Let $\Omega _{L}=\left\{ \omega :\sup Y\left( \omega \right) <t\right\} $, $\Omega _{U}=\left\{ \omega :\inf Y\left( \omega \right) >t\right\} $, and $\Omega _{M}=\Omega \setminus \left( \Omega _{L}\cup \Omega _{U}\right) $. By
definition, $\mathbf{P}\left( \Omega _{L}\right) < \tau $ and $\mathbf{P}\left( \Omega _{U}\right) < 1-\tau $. 
By $Sel(Y) \neq \emptyset$, we choose $y \in Sel(Y)$.
Let
\begin{equation*}
\tilde{y}\left( \omega \right) =
\begin{cases}
y(\omega) & \omega \in \Omega_L \cup \Omega_U
\\
t & \omega \in \Omega_M
\end{cases}
\end{equation*}
Since $Y\left( \omega \right) $ is convex, $t \in Y(\omega)$ for all $\omega \in \Omega_M$, and thus the random variable $\tilde y$ defined above is a selection of $Y$. 
By construction, $q_{\tilde{y}}\left( \tau \right) = t$ and thus $\left( \tilde{T%
}_{Y}^{-1}\left( \tau \right) ,\tilde{C}_{Y}^{-1}\left( \tau \right) \right)
\subset \Theta^Y_0\left( \tau \right) $.

Suppose that $\inf Y(\omega) > -\infty$ $\mathbf{P}$-a.s.
Then, since $Y(\omega)$ is closed, the random variable $\tilde y$ defined by $\tilde y(\omega) = \inf Y(\omega)$ is a selection of $Y$.
It also satisfies
$
T_Y^{-1}(\tau) 
= \inf\left\{ t : \mathbf{P}(Y \cap (-\infty,t] \neq \emptyset) \geq \tau \right\}
= \inf\left\{ t : F_{\tilde y}(t) \geq \tau \right\}.
$
Therefore, $\tilde T_Y^{-1}(\tau) = q_{\tilde y}(\tau) \in \Theta^Y_0(\tau)$.

Finally, suppose that $\sup Y(\omega) < \infty$ $\mathbf{P}$-a.s.
Then, since $Y(\omega)$ is closed, the random variable $\tilde y$ defined by $\tilde y(\omega) = \sup Y(\omega)$ is a selection of $Y$.
It also satisfies
$
C_Y^{-1}(\tau) 
= \inf\left\{ t : \mathbf{P}(Y \subset (-\infty,t] ) \geq \tau \right\}
= \inf\left\{ t : F_{\tilde y}(t) \geq \tau \right\}.
$
Therefore, $\tilde C_Y^{-1}(\tau) = q_{\tilde y}(\tau) \in \Theta^Y_0(\tau)$.
\end{proof}

Theorems \ref{prop:id} and \ref{prop:id_convex} together show that ${\Theta^Y_0(\tau)} = \left[q_{\inf Y}\left( \tau \right) ,q_{\sup Y}\left( \tau \right)\right]$ for all $\tau \in (0,1)$, if $Y$ has a selection and $Y(\omega)$ is a compact interval $\mathbf{P}$-a.s.
In other words, $\left[q_{\inf Y}\left( \tau \right) ,q_{\sup Y}\left( \tau \right)\right]$ is a sharp characterization of the identification set $\Theta^Y_0(\tau)$.
There are sufficient conditions that guarantees that $Sel(Y) \neq \emptyset$.
One such condition is that $Y$ is closed-valued and non-empty a.s., as stated in the Fundamental Selection Theorem \citep[][Theorem 2.13]{Molchanov05}.

Outcome variables which are reported as convex valued sets include several important cases that an empirical researcher may encounter. 
First is the case where $Y$ is generated by a response to a questioner that offers a distinct set of intervals to choose from. 
The conditions in Theorems \ref{prop:id} and \ref{prop:id_convex} are general enough to allows these intervals to be distinct, intersect or even be included in each other. 
A second type of data which is covered by these conditions is willingness to pay surveys. 
Contingent valuation surveys which employ the collapsing interval method are a prominent example for this case.
Sometimes in these surveys, the interval is indeed $[c,\infty)$ for some real $c$, and quantiles can be estimated while expectations cannot. 

Estimation of the sharp identified set $\left[q_{\inf Y}(\tau),q_{\sup Y}(\tau)\right]$ can be implemented simply by taking the sample $\tau$-th quantiles $\widehat q_{\inf Y}(\tau)$ and $\widehat q_{\sup Y}(\tau)$ of $\inf Y$ and $\sup Y$, respectively.
Inference can also be implemented by applying \citet{BeresteanuMolinari08} to the standard limit joint distribution of the empirical quantiles $\left(\widehat q_{\inf Y}(\tau),\widehat q_{\sup Y}(\tau)\right)$.

\section{Covariates}\label{sec:covariates}

Suppose that in addition to $Y$ we observe a vector of $p$ covariates, denoted by $x$.

\begin{assumption}[Data with Covariates]\label{a:conditional}
Let $(x,y^\ast,Y): \Omega \rightarrow \mathbb{R}^p \times \mathbb{R} \times \mathcal{K}(\mathbb{R})$ be such that 
$y^{\ast }$ is unobserved,
$(x,Y)$ is observed,
$Y(\omega)$ is non-empty $\mathbf{P}$-a.s.,
$y^{\ast }(\omega) \in Y(\omega),$ $\mathbf{P}$-a.s., and
the regular conditional probability measures of $Y$ and $y^\ast$ given $x$ exist.
\end{assumption}

The true response $y^\ast$ can arbitrarily depend on the covariates $x$ as long as the regular conditional probability measures exist.
Since $y^\ast$ is allowed to arbitrarily depend on $x$ and $y^\ast$ is contained in $Y$, the random set $Y$ generally depends on the covariates $x$ as well.

To account for the observed covariates, we use the following extended notations.
Let $F_{y|x}$ denote the conditional cumulative distribution function of $y$ given $x$. 
Let $q_{y|x}$ denote the conditional quantile function of $y$ given $x$.
In light of the regular conditional probability measures, the conditional containment functional $C_{Y|x}$ and conditional capacity functional $T_{Y|x}$ of $Y$ given $x$ are defined by 
\begin{align*}
C_{Y|x}(K) &= \mathbf{P}(\left\{ \omega : Y(\omega) \subset K, X(\omega)=\xi\right\}|\left\{ \omega : X(\omega)=\xi\right\})
\qquad\text{and}\\
T_{Y|x}(K) &= \mathbf{P}(\left\{ \omega : Y(\omega) \cap K \neq \emptyset, X(\omega)=\xi\right\}|\left\{ \omega : X(\omega)=\xi\right\}),
\end{align*}
respectively.

\subsection{Conditional Quantiles}\label{sec:nonparametric_conditional_quantile}

Theorems \ref{prop:id} and \ref{prop:id_convex} presented in Section \ref{sec:identification} naturally extend to conditional quantile counterparts.
Define the $\tau$-th conditional quantile set of $Y$ gien $x=\xi$ by
\begin{equation*}
\Theta^{Y|x}_0\left( \tau | \xi \right) =\left\{ q_{y|x}\left( \tau | \xi \right) :(x,y)\in Sel\left( x,Y\right) \right\}.
\end{equation*}
The following two corollaries are the extended counterparts of Theorems \ref{prop:id} and \ref{prop:id_convex}.

\begin{corollary}[Partial Identification]\label{corollary:id}
Suppose that Assumption \ref{a:conditional} holds. 
We have 
$\Theta^{Y|x}_0\left( \tau | \xi \right) \subset \left[ q_{\inf Y|x}\left( \tau |\xi\right) ,q_{\sup Y|x}\left( \tau |\xi\right) \right]$ for every $\tau \in \left( 0,1\right) $.
\end{corollary}

\begin{corollary}[Sharpness]\label{corollary:id_convex}
Suppose that Assumption \ref{a:conditional} holds. 
If $Y$ is a convex valued random set in $\mathcal{K}(\mathbb{R})$ with $Sel(x,Y) \neq \emptyset$, then $\left(q_{\inf Y|x}\left( \tau |\xi\right) ,q_{\sup Y|x}\left( \tau |\xi\right) \right) \subset \Theta^{Y|x}_0\left( \tau |\xi \right)$ for all $\tau \in \left( 0,1\right)$.
Furthermore, if in addition $\inf Y(\omega) \in Y(\omega) > -\infty$ for $\mathbf{P}$-a.s. $\omega \in \Omega$ then $q_{\inf Y|x}\left( \tau |\xi\right) \in \Theta^{Y|x}_0\left( \tau |\xi\right)$ for all $\tau \in \left( 0,1\right)$.
Similarly, if in addition $\sup Y(\omega) < \infty$ for $\mathbf{P}$-a.s. $\omega \in \Omega$, then $q_{\sup Y|x}\left( \tau|\xi \right) \in \Theta^{Y|x}_0\left( \tau |\xi\right)$ for all $\tau \in \left( 0,1\right)$.
\end{corollary}

Estimation of the sharp identified set $\left[q_{\inf Y|x}(\tau|\xi),q_{\sup Y|x}(\tau|\xi)\right]$ can be implemented by local polynomial estimators $\widehat q_{\inf Y|x}(\tau|\xi)$ and $\widehat q_{\sup Y|x}(\tau|\xi)$ for the $\tau$-th quantiles of $\inf Y$ and $\sup Y$, respectively, given $x=\xi$.
Inference can be also implemented by applying \citet{BeresteanuMolinari08} to the limit joint distribution of $\left(\widehat q_{\inf Y}(\tau),\widehat q_{\sup Y}(\tau)\right)$ based on a version of Bahadur representations \citep[e.g.,][]{Chaudhuri1991,guerre_sabbah2012,qu_yoon2015,qu_yoon2018}.

\subsection{Parametric Quantile Regression Models}\label{sec:quantile_regression}

We now turn to parametric quantile regression models.
We define a parametrized $\tau$-th quantile regression function $q( \ \cdot \ , \theta(\tau))$ by
\begin{equation*}
F_{y \mid x}( q(x,\theta(\tau)) \mid x) = \tau \qquad\text{$\mathbf{P}$-a.s.}
\end{equation*}
for all $\tau \in (0,1)$.
With the covariates $x$ allowed to include a constant,
a special case is the linear quantile regression function given by $q(x,\theta(\tau)) = x^\prime \theta(\tau)$.
We would like to identify $\theta(\tau)$ for each $\tau \in (0,1)$.
The identification set is defined by
\begin{equation}\label{eq:identified_set_parametric_quantile_regression}
\Theta_0(\tau) = \left\{ \theta(\tau) \in \Theta : F_{y \mid x}( q(x,\theta(\tau)) \mid x) = \tau \ \text{$\mathbf{P}$-a.s.}, (x,y) \in Sel(x,Y) \right\}
\end{equation}
for each $\tau \in (0,1)$.

\begin{theorem}[Partial Identification]\label{prop:quantile_regression:id}
Suppose that Assumption \ref{a:conditional} holds. 
For every $\tau \in (0,1)$,
\begin{equation}\label{eq:moment_inequalities}
\Theta_0(\tau) 
\subset \left\{ \theta(\tau) \in \Theta : 
\begin{array}{l}
E\left[ 1[Y \subset (-\infty,q(x,\theta(\tau))]] - \tau \mid x\right] \leq 0
\\
E\left[ \tau - 1[Y \cap (-\infty,q(x,\theta(\tau)]] \neq \emptyset] \mid x\right] \leq 0
\end{array}
\ \mathbf{P}-a.s.\right\}.
\end{equation}
\end{theorem}

\begin{proof}
As in the proof of Theorem \ref{prop:id}, $y \in Sel(Y \mid x=\xi )$ implies $\tilde C_{Y \mid x} (t \mid \xi) \leq F_{y \mid x}(t \mid \xi) \leq \tilde T_{Y \mid x} (t \mid \xi)$ for all $t \in \mathbb{R}$.
Therefore, $(x,y) \in Sel(x,Y)$ implies $\tilde C_{Y \mid x} (q(x,\theta(\tau)) \mid x) \leq F_{y \mid x}(q(x,\theta(\tau)) \mid x) \leq \tilde T_{Y \mid x} (q(x,\theta(\tau)) \mid x)$ $\mathbf{P}$-a.s.
Thus,
\begin{eqnarray*}
\Theta_0(\tau) &=& \left\{ \theta(\tau) \in \Theta : F_{y \mid x}( q(x,\theta(\tau)) \mid x) = \tau \ \text{$\mathbf{P}$-a.s.}, (x,y) \in Sel(x,Y) \right\}
\\
&\subset& \left\{ \theta(\tau) \in \Theta : \tilde C_{Y \mid x} (q(x,\theta(\tau)) \mid x) \leq \tau \leq \tilde T_{Y \mid x} (q(x,\theta(\tau)) \mid x) \ \mathbf{P}-a.s.\right\}.
\end{eqnarray*}
Writing the last expression in terms of conditional moments yields the expression in the statement of the theorem.
\end{proof}

This theorem provides conditional moment inequality restrictions to characterize a superset of the identification set $\Theta_0(\tau)$. 
Like Theorem \ref{prop:id} for quantiles, the other direction of set inclusion is not generally guaranteed.
However, like Theorem \ref{prop:id_convex} for quantiles, the following theorem provides an important case where the reverse inclusion is true.

\begin{theorem}[Sharpness]\label{prop:quantile_regression:id_convex}
Suppose that Assumption \ref{a:conditional} holds. 
If $Sel(x,Y) \neq \emptyset$ and $Y$ is a convex valued random set in $\mathbb{R}$ (i.e. $Y$ is interval valued), then
\begin{equation*}
\left\{ \theta(\tau) \in \Theta : 
\begin{array}{l}
E\left[ 1[Y \subset (-\infty,q(x,\theta(\tau))]] - \tau \mid x\right] < 0
\\
E\left[ \tau - 1[Y \cap (-\infty,q(x,\theta(\tau)]] \neq \emptyset] \mid x\right] < 0
\end{array}
\ \mathbf{P}-a.s.\right\}
\subset
\Theta_0(\tau)
\end{equation*}
for all $\tau \in (0,1)$.
Furthermore, if in addition $\inf Y(\omega) > -\infty$ $\mathbf{P}$-a.s., then
\begin{equation*}
\left\{ \theta(\tau) \in \Theta : 
\begin{array}{l}
E\left[ \tau - 1[Y \cap (-\infty,q(x,\theta(\tau)]] \neq \emptyset] \mid x\right] = 0
\end{array}
\ \mathbf{P}-a.s.\right\}
\subset
\Theta_0(\tau)
\end{equation*}
for all $\tau \in \left( 0,1\right)$.
Similarly, if in addition $\sup Y(\omega) < \infty$ $\mathbf{P}$-a.s., then 
\begin{equation*}
\left\{ \theta(\tau) \in \Theta : 
\begin{array}{l}
E\left[ 1[Y \subset (-\infty,q(x,\theta(\tau))]] - \tau \mid x\right] = 0
\end{array}
\ \mathbf{P}-a.s.\right\}
\subset
\Theta_0(\tau)
\end{equation*}
for all $\tau \in \left( 0,1\right)$.
\end{theorem}

\begin{proof}
If
$
\left\{ \theta(\tau) \in \Theta : \tilde C_{Y \mid x} (q(x,\theta(\tau)) \mid x) < \tau < \tilde T_{Y \mid x} (q(x,\theta(\tau)) \mid x) \ \mathbf{P}-a.s.\right\}
$
is empty, then the first claim in the theorem is satisfied.
Now, suppose that this set is non-empty.
Take
$
\theta \in
\left\{ \theta(\tau) \in \Theta : \tilde C_{Y \mid x} (q(x,\theta(\tau)) \mid x) < \tau < \tilde T_{Y \mid x} (q(x,\theta(\tau)) \mid x) \ \mathbf{P}-a.s.\right\}.
$
Let $\Omega_L = \{\omega : \sup Y(\omega) < q(x(\omega),\theta)\}$, $\Omega_U = \{\omega : \inf Y(\omega) > q(x(\omega),\theta)\}$, and $\Omega_M = \Omega \backslash (\Omega_L \cup \Omega_U)$.
Then, $\mathbf{P}(\Omega_L \mid x) < \tau$ and $\mathbf{P}(\Omega_U \mid x) < 1 - \tau$, $\mathbf{P}$-a.s.
By $Sel(x,Y) \neq \emptyset$, we choose $(x,y) \in Sel(x,Y)$.
Let
$$
\tilde y(\omega) =
\begin{cases}
y(\omega) & \omega \in \Omega_L \cup \Omega_U
\\
q(x(\omega),\theta) & \omega \in \Omega_M
\end{cases}
$$
Since $Y\left( \omega \right) $ is convex, $q(x(\omega),\theta) \in Y(\omega)$ for all $\omega \in \Omega_M$, and thus the random variable $\tilde y$ defined above is a selection of $Y$. 
Therefore, $(x,\tilde y) \in Sel(x,Y)$.
By construction, $F_{\tilde y \mid x}(q(x,\theta) \mid x) = \tau$, $\mathbf{P}$-a.s., and thus $\theta \in \Theta_0(\tau)$.
This shows 
$$
\left\{ \theta(\tau) \in \Theta : \tilde C_{Y \mid x} (q(x,\theta(\tau)) \mid x) < \tau < \tilde T_{Y \mid x} (q(x,\theta(\tau)) \mid x) \ \mathbf{P}-a.s.\right\} \subset \Theta_0(\tau).
$$
Writing the expression on the left-hand side in terms of conditional moments yields the expression in the statement of the theorem.

Suppose that $\inf Y(\omega) > -\infty$ $\mathbf{P}$-a.s.
Then, since $Y(\omega)$ is closed, the random variable $\tilde y$ defined by $\tilde y(\omega) = \inf Y(\omega)$ is a selection of $Y$.
Therefore, $(x,\tilde y) \in Sel(x,Y)$.
If we take $\theta \in \left\{ \theta(\tau) \in \Theta : \tau = \tilde T_{Y \mid x} (q(x,\theta(\tau)) \mid x) \ \mathbf{P}-a.s. \right\}$, then
$
\tau
=
\tilde T_{Y \mid x}(q(x,\theta) \mid x)
=
\mathbf{P}(Y \cap (-\infty, q(x,\theta)] \neq \emptyset \mid x)
= F_{\tilde y \mid x}(q(x,\theta) \mid x),
$
$\mathbf{P}$-a.s., and thus $\theta \in \Theta_0(\tau)$.
This shows 
$$
\left\{ \theta(\tau) \in \Theta : \tilde T_{Y \mid x} (q(x,\theta(\tau)) \mid x) =\tau \ \mathbf{P}-a.s.\right\} \subset \Theta_0(\tau).
$$
Writing the expression on the left-hand side in terms of conditional moments yields the expression in the statement of the theorem.

Suppose that $\sup Y(\omega) < \infty$ $\mathbf{P}$-a.s.
Then, since $Y(\omega)$ is closed, the random variable $\tilde y$ defined by $\tilde y(\omega) = \sup Y(\omega)$ is a selection of $Y$.
Therefore, $(x,\tilde y) \in Sel(x,Y)$.
If we take $\theta \in \left\{ \theta(\tau) \in \Theta : \tilde C_{Y \mid x} (q(x,\theta(\tau)) = \tau \mid x) \ \mathbf{P}-a.s. \right\}$, then
$
\tau
=
\tilde C_{Y \mid x}(q(x,\theta) \mid x)
=
\mathbf{P}(Y \subset (-\infty, q(x,\theta)] \mid x)
= F_{\tilde y \mid x}(q(x,\theta) \mid x),
$
$\mathbf{P}$-a.s., and thus $\theta \in \Theta_0(\tau)$.
This shows 
$$
\left\{ \theta(\tau) \in \Theta : \tilde C_{Y \mid x} (q(x,\theta(\tau)) \mid x) = \tau \ \mathbf{P}-a.s.\right\} \subset \Theta_0(\tau).
$$
Writing the expression on the left-hand side in terms of conditional moments yields the expression in the statement of the theorem.
\end{proof}

Theorems \ref{prop:quantile_regression:id} and \ref{prop:quantile_regression:id_convex} together show that the conditional moment inequality restrictions in (\ref{eq:moment_inequalities}) provide a sharp characterization of the identification set $\Theta_0(\tau)$ for all $\tau \in (0,1)$ if $(x,Y)$ has a selection and $Y(\omega)$ is a compact interval $\mathbf{P}$-a.s.
One sufficient condition for the condition, $Sel(x,Y) \neq \emptyset$, of Theorem \ref{prop:quantile_regression:id_convex} is that $Y$ is closed-valued and non-empty, as stated in the Fundamental Selection Theorem \citep[][Theorem 2.13]{Molchanov05}.
The conditional moment inequalities can be rewritten more simply as
\begin{equation}\label{eq:simple_moment_inequalities}
\Theta_0(\tau) =
\left\{ \theta(\tau) \in \Theta : 
\begin{array}{l}
E\left[ 1[y^u  \leq q(x,\theta(\tau))] - \tau \mid x\right] \leq 0
\\
E\left[ \tau - 1[y^l \leq q(x,\theta(\tau))] \mid x\right] \leq 0
\end{array}
\ \mathbf{P}-a.s.\right\}
\end{equation}
where $(x,y^l,y^u)$ is an $\mathbb{R}^{p+2}$-dimensional random vector generated by $(x,Y)$ through the transformation $(x(\omega),Y(\omega)) \mapsto (x(\omega),\min Y(\omega), \max Y(\omega)) =: (x(\omega),y^l(\omega),y^u(\omega))$ for each $\omega \in \Omega$.

For the quantile set, the sharp identification set is guaranteed to be an interval (see Section \ref{sec:identification}).
For the current setting where the sharp identification set is only implicitly characterized by a system of conditional moment inequalities, it is not clear if the identification set has nice geometric properties such as convexity.
Suppose that the quantile regression is specified in the linear-in-parameters form $q(x,\theta(\tau)) = x^\prime \theta(\tau)$.
In this case, the identification set $\Theta_0(\tau)$ can be shown to be convex.
Consequently, projections of the identification set $\Theta_0(\tau)$ on each coordinate is an interval.

\begin{theorem}[Convexity]
Suppose that Assumption \ref{a:conditional} holds. 
Suppose that $q(x,\theta(\tau)) = x^\prime \theta(\tau)$ for all $\tau \in (0,1)$.
If $Y$ is  a closed convex valued random set in $\mathbb{R}$, then the identification set $\Theta_0(\tau)$ is convex for all $\tau \in (0,1)$. 
\end{theorem}

\begin{proof}
Fix $\tau$.
By Theorems \ref{prop:quantile_regression:id} and \ref{prop:quantile_regression:id_convex}, the identification set $\Theta_0(\tau)$ is given by (\ref{eq:simple_moment_inequalities}) under the given conditions.
Let $\theta^1, \theta^2 \in \Theta_0(\tau)$ and $\lambda \in (0,1)$.
Then, $q(x,\theta(\tau)) = x^\prime \theta(\tau)$ implies $\{ \omega \in x^{-1}(\{\xi\}) \subset \Omega : y^u(\omega) \leq q(x(\omega),\lambda \theta^1 + (1-\lambda) \theta^2) \} \subset \{ \omega \in x^{-1}(\{\xi\}) \subset \Omega : y^u(\omega) \leq \max \{ q(x(\omega), \theta^1), q(x(\omega), \theta^2)\} \}$ for every $\xi \in \mathbb{R}^p$,
and thus
$$
E\left[ 1[y^u  \leq q(x,\lambda \theta^1 + (1-\lambda) \theta^2)] - \tau \mid x\right] \leq 0 \quad \mathbf{P}-a.s.
$$
Also, $q(x,\theta(\tau)) = x^\prime \theta(\tau)$ implies $\{ \omega \in x^{-1}(\{\xi\}) \subset \Omega : y^l(\omega) \leq q(x(\omega),\lambda \theta^1 + (1-\lambda) \theta^2) \} \supset \{ \omega \in x^{-1}(\{\xi\}) \subset \Omega : y^l(\omega) \leq \min \{ q(x(\omega), \theta^1), q(x(\omega), \theta^2)\} \}$ for every $\xi \in \mathbb{R}^p$,
and thus
$$
E\left[ \tau - 1[y^l \leq q(x,\lambda \theta^1 + (1-\lambda) \theta^2)] \mid x\right] \leq 0 \quad \mathbf{P}-a.s.
$$
Therefore, $\lambda \theta^1 + (1-\lambda) \theta^2 \in \Theta_0(\tau)$, showing that $\Theta_0(\tau)$ is convex.
\end{proof}

This geometric information is useful in practice.
For example, it provides a guidance about the direction of computational search for a grid representation of set estimates.
Furthermore, this result guarantees that a projection of the identification set is an interval, which will be useful when methods of inference becomes available in the future for projections of identified sets under conditional moment inequalities.\footnote{See for example \cite{BelloniBugniChernozhukov2018} and \cite{BugniShi2018}. Also see \citet{KaidoMolinariStoye2016}.}

The conditional moment inequality restrictions in (\ref{eq:simple_moment_inequalities}) to characterize the identification set $\Theta_0(\tau)$ can be rewritten as
$$
E\left[ m_j(w,\theta) \mid x \right] \geq 0 \quad\mathbf{P}-a.s. \quad\text{for } j = 1,2,
$$
where the moment functions, $m_1$ and $m_2$, are defined by
\begin{eqnarray*}
m_1(w,\theta) &=& 1[y^l \leq q(x,\theta)] - \tau \qquad\text{and}
\\
m_2(w,\theta) &=& \tau - 1[y^u \leq q(x,\theta)]
\end{eqnarray*}
for $w = (x,y^l,y^u)$.
This model fits in the framework for which the existing literature provides methods of inference via moment selection, e.g., \citet{AndrewsShi2013}.
For convenience of the readers, we describe the procedure of inference based on this existing literature in Appendix \ref{sec:procedure_inference}. Furthermore, Appendix \ref{sec:best_linear} provides a practical procedure to supplement this inference procedure in the context of linear quantile regressions.

\section{Simulation Studies}\label{sec:simulation_studies}

In this section, we use simulation studies to evaluate the proposed method.
We consider the identified set $\Theta_0(\tau)$ of quantile regression parameters defined in (\ref{eq:identified_set_parametric_quantile_regression}), and conduct inference for the true parameters based on the moment inequalities (\ref{eq:simple_moment_inequalities}) that sharply characterize this identified set.
The method of inference outlined in Appendix \ref{sec:procedure_inference} is applied to these moment inequalities (\ref{eq:simple_moment_inequalities}).
With this method of inference and our characterization of the identified set, we use Monte Carlo simulations to check the size control at the true parameter value contained in the identified set.

In order to generate analytically tractable parametric quantile regressions, we use the following model.
\begin{eqnarray}
Y_i &=& \sum_{t=-\infty}^\infty 1[y_i \in (t-0.1,t]] \cdot [t-0.1,t],
\qquad\text{where}
\nonumber
\\
y_i &=& 1.0 + (1.0 + x_i) \cdot \varepsilon_i,
\label{eq:mc_parametric}
\\
x_i &\sim& U(0,1), \qquad \varepsilon_i \ \sim \ U(0,1), \qquad \text{and} \ \quad x_i \perp\!\!\!\perp \varepsilon_i.
\nonumber
\end{eqnarray}
For example, if $x_i = 0.385$ and $\varepsilon_i = 0.573$, then $y_i \approx 1.794$ and $Y_i = [1.7,1.8]$.
This individual $i$ is at the ($\tau = \varepsilon =$) 0.573-th quantile, and the corresponding quantile regression function is given by $q(x,\theta(0.573)) = 1.573 + 0.573 x$.
In fact, model (\ref{eq:mc_parametric}) yields the linear model for each quantile $\tau \in (0,1)$ with intercept $\theta_1(\tau) = 1+\tau$ and slope $\theta_2(\tau) = \tau$.
Hence, the true quantile regression parameter vector is $\theta(\tau) = (\theta_1(\tau),\theta_2(\tau)) = (1+\tau,\tau)$. 
Our identification theory shows that the moment inequalities (\ref{eq:simple_moment_inequalities}) give the sharp characterization of the identified set $\Theta_0(\tau)$ containing this true parameter vector, and an application of the existing inference method outlined in Appendix \ref{sec:procedure_inference} thus allows for inference about this true parameter vector.

For Monte Carlo experiments, 1,000 small samples are drawn of sizes 100 and 200 observations.
Based on the methods presented in Section \ref{sec:quantile_regression}, we conduct tests of $H_0: \theta(\tau) = \theta$ versus $H_1: \theta(\tau) \neq \theta$ for various parameter vectors $\theta$ for each quantile $\tau \in \{0.25, 0.50, 0.75\}$.
The critical value is estimated using the procedure outlined in Appendix \ref{sec:procedure_inference} with 1,000 bootstrap replications for the size $\alpha = 0.05$.\footnote{We used $R=2$, $\epsilon_n=0.05$, $\kappa_n=(0.3\ln(n))^{\frac{1}{2}}$ and $ B_n = (0.4\ln(n)/\ln\ln(n))^{\frac{1}{2}}$.}

Figure \ref{fig:mc_parametic} shows graphs of rejection frequencies for the test described above.
Parts (I), (II), and (III) of the figure focus on quantiles $\tau=$ 0.25, 0.50, and 0.75, respectively.
For ease of presentation, we focus on the one-dimensional slice of the two-dimensional parameter space running through the true parameter point at each quantile $\tau$.
Specifically, for the $\tau$-th quantile, the results are displayed on the set $\{(\theta_1(\tau),\theta_2(\tau)) \in \mathbb{R}^2 : \theta_1(\tau) = 1+ \tau, \theta_2(\tau) \in [-0.5,1.5] \}$, and the horizontal axes in Figure \ref{fig:mc_parametic} represent this set indexed by $\theta_2(\tau) \in [-0.5,1.5]$. 
The true point is indicated in the figure by a solid vertical line.
The rejection frequencies are a indicated by dashed and dotted curves for sample sizes $n=100$ and 200, respectively.
The uniform asymptotic theory of \citet[][Section 5]{AndrewsShi2013} together with our identification theory predicts that the implied size control entails the rejection probabilities asymptotically at or below the nominal rejection probability in the identified set. The curves displayed in the figure running at or below the nominal rejection probabilities are thus consistent with the theory.

\section{Empirical Application}\label{sec:empirical}

To assess the effects of localized environmental benefits on housing values through an extended hedonic regression model, \cite[][henceforth GRT]{GamperTimmins13} use within-tract quantiles of housing values as outcome variables.
Controlling for tract fixed effects through first differencing as well as controlling for various observed tract attributes, they find that localized cleanup of hazardous waste sites leads to larger appreciation in house prices at lower quantiles.
The Decennial Census provides house prices only in terms of intervals.
GRT used these intervals to compute quantile points of house prices -- see Table A.1 of their online appendix.

In this paper, we re-examine the empirical analysis of GRT by treating the housing prices as interval-valued data as in the original data, rather than adding assumptions to obtain point-valued housing prices.
Our econometric procedure consists of three steps.
First, we obtain within-tract quantile sets $\hat\Theta^Y_{i,t}(\tau)$ for each tract $i \in \{1,\ldots,n\}$ for each year $t \in \{90,00\}$ by using the methods proposed in Appendix \ref{sec:discrete}.
Due to the discrete nature of the house price intervals, we can rely on the super-consistency property presented in Appendix \ref{sec:discrete} -- see Theorem \ref{prop:discrete}.
In the second step, we take the set difference $\Delta \hat\Theta^Y_{i}(\tau) = \hat\Theta^Y_{i,00}(\tau) - \hat\Theta^Y_{i,90}(\tau)$ and take first differences $\Delta x^i$ of all the observed covariates as in GRT.
This second step is used to difference out tract fixed effects.
Finally, we use \cite[Sec. 4]{BeresteanuMolinari08} and the constructed data $(\Delta x^i, \Delta \hat\Theta^Y_{i}(\tau))$ to obtain set-valued best linear predictors of $\Delta x^i$ on $\Delta \hat\Theta^Y_{i}(\tau)$.
In the asymptotic setting where the number of within-tract houses increases as the number $n$ of tracts increases, the effects of the first-step quantile set estimation on the subsequent set estimation can be ignored due to the super-consistency result for the quantile set estimation under discrete categories (see Theorem \ref{prop:discrete} in Appendix \ref{sec:superconsistency}).

We use the same set of observed controls $\Delta x^i$ as in GRT.
In Table \ref{tab:empirical_results}, we report the estimates for the coefficient of cleanup of hazardous waste sites.
The table consists of nine rows showing estimation results across the nine quantiles $\tau \in \{0.1,\ldots,0.9\}$.
Column (I) shows point estimates of GRT as a reference.\footnote{We remark that the sample used by GRT and the sample used by us are slightly different likely due to different data selection procedures. Furthermore, they use sampling weights, while we do not.}
In column (II), we assume that $y$ equals the lower bound $a$ of interval categories $Y=[a,b]$, and obtained point estimates.
Similarly, in column (III) we assume that $y$ equals the upper bound $b$ of interval categories $Y=[a,b]$, and obtained point estimates.
Column (IV) shows the set estimates according to the three-step procedure outlined in the previous paragraph.
Column (V) in addition shows the 95\% confidence region to account for the sampling variations.

For most quantiles $\tau \in \{0.1,\ldots,0.7\}$, the numbers in column (I) are located around the middle between the numbers shown in columns (II) and (III).
Furthermore, those numbers in column (I) are also located around the middle of the set estimates shown in column (IV) and the 95\% confidence regions shown in column (V).
Column (I) indicates that GRT obtain significant effects across all the quantiles.
Column (II) implies that the significant positivity may disappear, depending on which selection of the interval-valued outcomes we take.
Column (III) implies that we may get even larger and more significant effects than those reported by GRT, depending on which selection of the interval-valued outcomes we take.
Columns (IV) and (V) show that the range of effects that we could get out of all the potential selections in the interval-valued outcome is large, and of course contains those numbers reported in columns (II) and (III) in particular.
Previous results depend on the assumptions made on the selection mechanism (e.g., taking the middle point).
Our analysis shows that different assumptions on the selection can lead to quite different results.

\section{Summary}

This paper investigates the identification of quantiles and of quantile regressions when the outcome variable is interval valued.
We first identify the quantiles of a random set. 
The quantile set is shown to be identified by the containment and capacity quantiles.
The sharpness of the identification set was shown when the set-valued outcomes are interval-valued.
The identification argument for quantiles then is extended to set identification of quantile regression parameters.
We show that a system of conditional moment inequalities, involving the cumulative containment and cumulative capacity functionals, characterize the identification set for quantile regression parameters.
The sharpness of the identification set was shown when the set-valued outcomes are interval-valued.

\newpage
\section*{Appendix}
\appendix
\section{Procedure of Inference for Quantile Regression Parameters}\label{sec:procedure_inference}

This appendix section provides the procedure of inference for quantile regression parameters based on the conditional moment inequality restrictions derived in Section \ref{sec:quantile_regression} via the method of \citet{AndrewsShi2013}.
Recall that, if $Y$ is a closed convex value random set, then identification set $\Theta_0(\tau)$ of the quantile regression parameters $\theta$ is characterized by the conditional moment inequalities
$$
E\left[ m_j(w,\theta) \mid x \right] \geq 0 \quad\mathbf{P}-a.s. \quad\text{for } j = 1,2,
$$
where the moment functions, $m_1$ and $m_2$, are defined by
\begin{eqnarray*}
m_1(w,\theta) &=& 1[y^l \leq q(x,\theta)] - \tau \qquad\text{and}
\\
m_2(w,\theta) &=& \tau - 1[y^u \leq q(x,\theta)]
\end{eqnarray*}
for $w = (x,y^l,y^u)$.
Recall also that this set is convex if the quantile regression function $q( \ \cdot \ , \theta)$ is linear in parameters $\theta$.

Normalize the vector $x$ of $p$ covariates into $[0,1]^p$, and define the sample moment functions and the sample variances by
\begin{eqnarray*}
\bar m_{nj}(\theta,g) &=& n^{-1} \sum_{i=1}^n m_j(w_i,\theta) \cdot g(x_i) \quad\text{for } j=1,2 \text{ and}
\\
\hat\sigma_{nj}^2(\theta,g) &=& n^{-1} \sum_{i=1}^n \left( m_j(w_i,\theta) \cdot g(x_i) - \bar m_{nj}(\theta,g) \right)^2 \quad\text{for }j=1,2,
\end{eqnarray*}
for a function $g$ to be defined below.
To bound the sample variance away from zero, we use
$$
\bar \sigma_{nj}^2(\theta,g) = \hat\sigma_{nj}^2(\theta,g) + \epsilon_n \cdot \hat\sigma_{nj}^2(\theta,1) \quad\text{for }j=1,2.
$$
With $g_{a,r}(x) = 1\left[ x \in \prod_{u=1}^p \left( \frac{a_u-1}{2r}, \frac{a_u}{2r} \right] \right]$, an approximated test statistic at $\theta$ is computed by
\begin{equation*}
\bar T_{nR}(\theta) = \sum_{r=1}^R (r^2 + 100)^{-1} \sum_{a \in [1,\cdots,2r]^p} (2r)^{-p} \left( \left[\frac{n^{\frac{1}{2}} \bar m_{n1}(\theta,g_{a,r})}{\bar \sigma_{n1}(\theta,g_{a,r})}\right]_-^2 + \left[\frac{n^{\frac{1}{2}} \bar m_{n2}(\theta,g_{a,r})}{\bar \sigma_{n2}(\theta,g_{a,r})}\right]_-^2 \right)
\end{equation*}
for some truncation number $R \in \mathbb{N}$, where $[x]_- = -x$ if $x <0$ and $[x]_-=0$ if $x \geq 0$.

Lemma 1 of \citet{AndrewsShi2013} guarantees that our definition of $\bar T_{nR}(\theta)$ satisfies Assumptions S1--S4 in their paper.
Likewise, Lemma 3 of \citet{AndrewsShi2013} guarantees that the choice of $g_{a,r}$ defined above satisfies Assumptions CI and M in their paper.
In order to assure that we can use the method of \citet{AndrewsShi2013}, it remains to check their condition (2.3).
The following conditions suffice:
$w$ is i.i.d.;
$0 < \text{Var}(m_1(w,\theta)) < \infty$;
$0 < \text{Var}(m_2(w,\theta)) < \infty$;
$E | m_1(w,\theta) / \sigma_1(\theta) |^{2+\delta} < \infty$; and
$E | m_2(w,\theta) / \sigma_2(\theta) |^{2+\delta} < \infty$;
where $\delta > 0$, $\sigma_1(\theta) = \text{Var}(m_1(w,\theta))$, and $\sigma_2(\theta) = \text{Var}(m_2(w,\theta))$.

To compute the critical value for $\bar T_n(\theta)$, generate $B$ bootstrap samples $\{w_{ib}^\ast : i = 1,\cdots,n\}$ for $b=1,\cdots,B$.
For each bootstrap sample $\{w_{ib}^\ast : i = 1,\cdots,n\}$, compute $\bar m_{nbj}^\ast(\theta,g)$ and $\bar \sigma_{nbj}^\ast(\theta,g)$ for $j=1,2$.
For each bootstrap sample, compute the bootstrap test statistic\small
\begin{eqnarray*}
\bar T_{nbR}^\ast(\theta) = \sum_{r=1}^R (r^2 + 100)^{-1} \sum_{a \in [1,\cdots,2r]^p} (2r)^{-p} \left( \left[\frac{n^{\frac{1}{2}} (\bar m_{nb1}^\ast(\theta,g_{a,r}) - \bar m_{n1}(\theta,g_{a,r})) / \hat\sigma_{n1}(\theta,1) + \varphi_n(\theta,g_{a,r})}{\bar \sigma_{nb1}^\ast(\theta,g_{a,r}) / \hat\sigma_{n1}(\theta,1)}\right]_-^2 \right.
\\
\left. + \left[\frac{n^{\frac{1}{2}} (\bar m_{n2}(\theta,g_{a,r}) - \bar m_{n2}(\theta,g_{a,r})) / \hat\sigma_{n2}(\theta,1) + \varphi_n(\theta,g_{a,r})}{\bar \sigma_{nb2}^\ast(\theta,g_{a,r}) / \hat\sigma_{n2}(\theta,1)}\right]_-^2 \right)
\end{eqnarray*}\normalsize
where $\varphi_{nj}(\theta,g)$ is given by
$$
\varphi_{nj}(\theta,g) = B_n 1[ \kappa_n^{-1} n^{\frac{1}{2}} \bar m_{nj}(\theta,g) / \bar\sigma_{nj}(\theta,g) > 1]
\quad\text{for }j=1,2.
$$
\citet{AndrewsShi2013} recommend $\epsilon_n = 0.05$, $\kappa_n =(0.3 \ln(n))^{\frac{1}{2}}$, and $B_n = (0.4 \ln(n) / \ln\ln(n))^{\frac{1}{2}}$.
The critical value $\bar c_{nRB,1-\alpha}^\ast(\theta)$ is set to be the $1-\alpha+10^6$ sample quantile of the bootstrap test statistics.
Thus, a nominal level $1-\alpha$ confidence set is approximated by
$$
\left\{ \theta \in \Theta : \bar T_{nR}(\theta) \leq \bar c_{nRB,1-\alpha}^\ast(\theta) \right\}.
$$
Because finding the approximate region for this set is computationally burdensome when the dimension of the parameter set $\Theta$ is large, we provide in Appendix \ref{sec:best_linear} an efficient algorithm to compute the estimate of the identification set.

\section{Set of Best Linear Predictors}\label{sec:best_linear}

In this section, we present a novel set programming method for ease of implementation of obtaining set estimate for linear quantile regressions, as motivated for practical ease fo finding the approximate region at the end of Appendix \ref{sec:procedure_inference}.
We focus on the linear quantile regression model, $q(x,\theta) = x^\prime \theta$. 
We define the identification region $\Theta^\ast_0(\tau)$ for the best linear predictors (BLP) $\theta$ by extending \cite{KoenkerBasset78} for the case of interval valued $Y$. 
In addition, we show in Theorem \ref{prop:blp_containment} that the identification region $\Theta^\ast_0(\tau)$ for the BLP model is a superset of the identification region $\Theta_0(\tau)$ defined in (\ref{eq:simple_moment_inequalities}) for linear models.
This result is important since finding the identification $\Theta_0(\tau)$ defined in (\ref{eq:simple_moment_inequalities}) is challenging relatively to finding the identification region $\Theta^\ast_0(\tau)$ for the BLP model. 
Therefore, one can start by finding the identification region for the BLP model and use this superset as a starting region where we should look for the identification region of Section \ref{sec:quantile_regression}.

For a given selection $(x,y) \in Sel(x,Y)$, we can estimate the associated parameters by minimizing the risk
\begin{equation}
R_\tau\left( \theta ;x,y\right) =E\left[ \rho_\tau \left( y-x^\prime \theta \right) \right]
\label{LossFunction}
\end{equation}
where $\rho_\tau(u) = u \cdot (\tau - 1[u < 0])$.
(See \cite{KoenkerBasset78} and \cite{Koenker05}).
We propose that the identification set $\Theta_0(\tau)$ be approximated by
\begin{equation*}
\Theta _0^\ast (\tau) =\left\{ \arg \min_{\theta \in \Theta} R_\tau\left( \theta ;x,y\right) : (x,y) \in Sel\left(x, Y\right) \right\} .
\end{equation*}
In appendix \ref{sec:geometric_blp}, we provide some useful geometric properties of this set of best linear predictors.
More importantly, we show in the theorem below that it is useful to locate the identification set $\Theta_0(\tau)$.

\begin{theorem}\label{prop:blp_containment}
If $q(x,\theta) = x' \theta$, then
$\Theta_0(\tau) \subset \Theta_0^\ast(\tau)$.
\end{theorem}

\begin{proof}
Suppose that $\theta \in \Theta_0(\tau)$.
In other words, $F_{y \mid x}(x^\prime \theta \mid x) = \tau$ $\mathbf{P}$-a.s. for a selection $(x,y) \in Sel(x,Y)$.
Taking the gradient of $R_\tau\left( \theta ;x,y\right)$ with respect to $\theta$, we have
\begin{eqnarray*}
\nabla_\theta \int_{\mathbb{R}^p} \left[ (\tau-1) \int_{-\infty}^{\xi^\prime \theta} (\zeta - \xi^\prime \theta) dF_{y \mid x}(\zeta \mid \xi) + \tau \int_{\xi^\prime \theta}^\infty (\zeta - \xi^\prime \theta) dF_{y \mid x}(\zeta \mid \xi) \right] dF_x(\xi)
\\
\ = \
\int_{\mathbb{R}^p} \xi \left[F_{y \mid x}(\xi^\prime \theta \mid \xi) - \tau \right] dF_x(\xi)
\ = \
0
\end{eqnarray*}
where the last equality follows from our choice of $\theta$ satisfying $F_{y \mid x}(x^\prime \theta \mid x) = \tau$ $\mathbf{P}$-a.s.
Therefore, we obtain $\nabla_\theta R(\theta; x,y) = 0$.
Since $R_\tau( \ \cdot \ ; x,y)$ is convex, this implies $\theta \in \Theta_0^\ast(\tau)$.
\end{proof}

The other direction of set inclusion does not hold.
While the identification set $\Theta_0(\tau)$ contains only those parameter vectors $\theta$ that correctly specify the quantile regression $q(x,\theta) = x^\prime \theta$ for some $(x,y) \in Sel(x,Y)$, the approximate set $\Theta_0^\ast(\tau)$ contains many other parameter vectors $\theta$ which only allow $q(x,\theta) = x^\prime \theta$ to be a best linear predictor for some $(x,y) \in Sel(x,Y)$.
By Theorems \ref{prop:blp_containment}, the approximation set $\Theta_0^\ast(\tau)$ does not miss any element of the identification set $\Theta_0(\tau)$.
Therefore, we propose to first compute this set $\Theta_0^\ast(\tau)$ of best linear predictors, and conduct the test of conditional moment inequalities on and around this set.
If the identification set $\Theta_0(\tau)$ is empty, then the parametric quantile regression model is misspecified, and $\Theta_0^\ast(\tau)$ trivially contains $\Theta_0(\tau)$.
In this case of misspecification, the set $\Theta_0^\ast(\tau)$ of best linear predictors itself may be of use for best linear prediction and for causal inference -- see \cite{AngristChernozhukovFernandezval2006} and \cite{KatoSasaki2017}.
The remainder of this section is devoted to a computational algorithm to obtain the approximation set $\Theta_0^\ast(\tau)$.

\subsection{Linear Programming}

Since a random set can be viewed as a collection of regular random variables, we start by reviewing the regular random variable case.
It seems logical that what we define for random sets should yield a regular quantile regression for singleton random sets.

The canonical LP problem is written as the following constrained
minimization problem, 
\begin{eqnarray}
&\min& c^{\prime }\theta  \label{canonical} \\
&s.t& A\theta  = b  \notag \\
&& \theta \geq 0,  \notag
\end{eqnarray}
where $A$ is a $m\times k$ matrix, 
$c$ is a $k\times 1$ vector of coefficients, 
$b$ is a $m\times 1$ vector of right-hand side constraints, and 
$\theta $ is a $k\times 1$ vector of unknowns. 
It is assumed that $m<k$.

Consider a finite random sample of $n$ observations.
Let $y$ be the $n\times 1$ vector of outcomes, $X$ be the $n\times p$ matrix of covariates (which includes a column of ones).
We can solve the least absolute deviation problem corresponding to the $\tau $-quantile regression by using the following linear programming problem,
\begin{eqnarray}
&&\min_{\left( \beta ,u,v\right) \in 
\mathbb{R}
^{p}\times 
\mathbb{R}
_{+}^{2n}}\sum_{i=1}^{n}\tau u_{i}+\left( 1-\tau \right) v_{i}  \label{LP} \\
&&s.t.\ \ \ X\beta +u-v=y.  \notag
\end{eqnarray}
The vector $( \beta _{j}) _{j=1}^{p}$ consists of the coefficients of the linear $\tau $-quantile regression, while $u$ and $v$ are slack parameters (variables). 
The LP problem in (\ref{LP}) is labeled $LP\left( \tau ,X,y\right) $. 
The simplex algorithm provides a solution to the above problem. 
The first stage is to transform the linear programming problem in (\ref{LP}) into the canonical form. 
Note that (\ref{canonical}) requires that all variables over which we minimize be
positive while the coefficients $\beta $ in (\ref{LP}) are unrestricted. 
The first step requires the user to transform the problem into the form in (\ref{canonical}). 
So at first, we can write (\ref{LP}) as 
\begin{eqnarray}
&\min& c^{\prime }x  \label{almost_canonical} \\
&s.t.& Ax = b(y)  \notag \\
&&x \in S  \notag
\end{eqnarray}%
where $c=\left( 0_{p};\tau 1_{n};\left( 1-\tau \right) 1_{n}\right)' $, 
$x=\left( \beta ,u,v\right)' $, 
$A=[X:I_{n\times n},-I_{n\times n}]$, 
$b(y)=y$, and 
$S=\mathbb{R}^{p}\times \mathbb{R}_{+}^{2n}$. 
The first $p$ coordinates of $x$ are unrestricted while the last $2n$ coordinates are restricted to be non-negative. 
Some software packages handle this kind of almost canonical form (e.g. Matlab) but some more traditional code may not. 
Assume w.l.o.g. that the first $p$ rows of the matrix $X$ form a $p\times p$ full rank (and thus invertible) matrix. 
Denote this matrix by $X_{p}$ and similarly denote by $u_{p}$, $v_{p}$ and $y_{p}$
the first $p$ lines of the the corresponding column vectors $u$, $v$ and $y$. 
Denote by $X_{-p}$, $u_{-p}$, $v_{-p}$ and $y_{-p}$ the remaining $n-p$ rows of these matrices and vectors. 
The first $p$ equations in $Ax=b$ as well as the unconstrained variable $\beta $ can be eliminated by writing $\beta =X_{p}^{-1}\left( y_{p}-u_{p}-v_{p}\right) $ and substituting $\beta $ into these $p$ first equations in $Ax=b$. 
The remaining $n-p$ equations then can be written as 
\begin{equation*}
\left( X_{-p}X_{p}^{-1}\right) u_{p}-\left( X_{-p}X_{p}^{-1}\right)
v_{p}+u_{-p}-v_{-p}=y_{-p}-\left( X_{-p}X_{p}^{-1}\right) y_{p}.
\end{equation*}%
Therefore, the problem in (\ref{almost_canonical}) can be written as
\begin{eqnarray}
&\min & c^{\prime }\tilde{x}  \label{canonical2} \\
&s.t.& \tilde{A}\tilde{x} =\tilde{b}(y)  \notag \\
&&\tilde{x} \geq 0  \notag
\end{eqnarray}
where 
$\tilde{c}=\left( \tau 1_{n},\left( 1-\tau \right) 1_{n}\right) $, 
$\tilde{x} =\left( u,v\right) $, $\tilde{A}=\left[ X_{-p}X_{p}^{-1}:I_{n-p\times n-p}:X_{-p}X_{p}^{-1}:I_{n-p\times n-p}\right] $,
and $\tilde{b}(y)=y_{-p}-\left( X_{-p}X_{p}^{-1}\right) y_{p}$. 
Notice that the LP problem in (\ref{LP}) has $p+2n$ variables and $n$ equality constraints and the LP problem in (\ref{canonical2}) has $2n$ variables and $n-p$ equality constraints. 
Out of the solution for $u$ and $v$ in (\ref{canonical2}) we can of course recover the vector of interest $\beta $ by using $\beta =X_{p}^{-1}\left( y_{p}-u_{p}-v_{p}\right) $.

For any ordered set $B$, let $B(i)$ denote the $i$-th element of $B$.
For any ordered set $B \subset \{1,\cdots, 2n\}$ of cardinality $|B|$ and for any $|B|$-dimensional vector $\xi$, we define the $2n$-dimensional vector $\Pi_B \xi$ whose $j$-th coordinate is given by
\begin{equation*}
(\Pi_B \xi)_j = \begin{cases}
\xi_i & \text{if there exists } i \in \{1,\cdots,|B|\} \text{ such that } B(i) = j \\
0     & \text{otherwise}
\end{cases}
\end{equation*}
for each $i \in \{1,\cdots,2n\}$.
The simplex algorithm yields a solution $\tilde x(y)$ with an ordered set $B_{y} \subset \{1,\cdots,2n\}$ of basic indices.
Also let $-B_{y} = \{1,\cdots,2n\} \backslash B_{y}$ denote an ordered set of non-basic indices.
It is required that $|B_{y}| = n-p$, $|-B_{y}| = n+p$, $\tilde x(y)_{B_{y}} \geqslant 0_{n-p}$, and $\tilde x(y)_{-B_{y}} = 0_{n+p}$.
Using these notations, the solution is explicitly written as the $2n$-dimensional vector
$\tilde x(y) = \Pi_{B_{y}} \tilde A^{-1}_{B_{y}} \tilde b(y)$.
Optimality requires $c_{-B_y} - c_{B_y} \tilde A_{B_y}^{-1} \tilde A_{-B_y} \geqslant 0_{n+p}$ and feasibility requires that $\tilde x(y) = \Pi_{B_{y}} \tilde A^{-1}_{B_{y}} \tilde b(y) \geqslant 0_{2n}$.
The simplex algorithm prescribes an efficient computational procedure to find such an index set $B_y$ satisfying these requirements.

\subsection{Set Linear Programming}

A simple brute force approach to set estimation of $\beta$ is to obtain the solution $\tilde x(y) = (u(y),v(y))$ to the optimization problem (\ref{canonical2}) for each $y \in \times_{i=1}^n [y_{Li}, y_{Ui}]$, and to take the union $\cup_{y \in \times_{i=1}^n [y_{Li}, y_{Ui}]} \beta(u)$ where $\beta(y) =X_{p}^{-1}\left( y_{p}-u_{p}(y)-v_{p}(y)\right) $.
However, this exhaustive approach (even with a lattice approximation) is computationally intensive.
In this light, we use some convenient properties of linear programming to propose a faster algorithm to compute the set estimate for $\beta$.

Pick $y^1 \in \times_{i=1}^n [y_{Li}, y_{Ui}]$.
The simplex algorithm yields a solution $\tilde x(y^1)$ with an ordered set $B_{y^1} \subset \{1,\cdots,2n\}$ of basic indices.
Also let $-B_{y^1} = \{1,\cdots,2n\} \backslash B_{y^1}$ denote an ordered set of non-basic indices.
The solution is explicitly written as the $2n$-dimensional vector
$\tilde x(y^1) = \Pi_{B_{y^1}} \tilde A^{-1}_{B_{y^1}} \tilde b(y^1)$.
The next proposition shows that there is a set $\mathcal{Y}^1 \in \times_{i=1}^n [y_{Li},y_{Ui}]$ containing $y^1$ such that the set $B_{y^1}$ of basic indices for (\ref{canonical2}) remains unchanged for all $y \in \mathcal{Y}^1$.
Therefore, once we solve (\ref{canonical2}) for $y = y^1$, we do not need to solve (\ref{canonical2}) again for any other $y \in \mathcal{Y}^1$, and we thus tremendously save our computational resources.

\begin{proposition}
A solution to (\ref{canonical2}) is given by
$\tilde x(y) = \Pi_{B_{y^1}} \tilde A^{-1}_{B_{y^1}} \tilde b(y)$
for all $y \in \mathcal{Y}^1$ where
$\mathcal{Y}^1 = \left\{ y \in \times_{i=1}^n [y_{Li}, y_{Ui}] \left\vert \tilde A^{-1}_{B_{y^1}} \tilde b(y) \geqslant 0_{n-p} \right.\right\}$.
In particular, $y^1 \in \mathcal{Y}^1$.
\end{proposition}

\begin{proof}
We have $c_{-B_{y^1}} - c_{B_{y^1}} \tilde A_{B_{y^1}}^{-1} \tilde A_{-B_{y^1}} \geqslant 0_{n+p}$ by the definition of $B_{y^1}$ and $-B_{y^1}$ as the sets of basic and non-basic indices, respectively, at the solution to (\ref{canonical2}) with $y=y^1$.
Notice that $c$ does not depend on $y$ in (\ref{canonical2}).
Hence, any feasible vertex $\tilde x \in \mathbb{R}^{2n}_+$ of the constraint set having non-zero elements only for those indices in $B(y^1)$ is optimal for (\ref{canonical2}).
Consider (\ref{canonical2}) with $y \in \times_{i=1}^n [y_{Li},y_{Ui}]$.
A vertex $\tilde x$ having non-zero elements only for those indices in $B(y^1)$ is written as $\tilde x = \Pi_{B_{y^1}} \tilde A_{B_{y^1}}^{-1} b(y)$.
It is feasible if $\tilde x = \Pi_{B_{y^1}} \tilde A_{B_{y^1}}^{-1} b(y) \geqslant 0_{2n}$, which is true if and only if $\tilde A_{B_{y^1}}^{-1} b(y) \geqslant 0_{n-p}$.
Therefore, a solution to (\ref{canonical2}) is given by
$\tilde x(y) = \Pi_{B_{y^1}} \tilde A^{-1}_{B_{y^1}} \tilde b(y)$
for all $y \in \times_{i=1}^n [y_{Li}, y_{Ui}]$ such that $\tilde A^{-1}_{B_{y^1}} \tilde b(y) \geqslant 0_{n-p}$.
\end{proof}

In light of this proposition, we propose the following procedure.
For any $y \in \mathcal{Y}^1$, let 
\begin{eqnarray*}
u_p(y) &=& \left(\left( \Pi_{B_{y^1}} \tilde A^{-1}_{B_{y^1}} \tilde b(y) \right)_1, \cdots, \left( \Pi_{B_{y^1}} \tilde A^{-1}_{B_{y^1}} \tilde b(y) \right)_p\right)^\prime
\\
v_p(y) &=& \left(\left( \Pi_{B_{y^1}} \tilde A^{-1}_{B_{y^1}} \tilde b(y) \right)_{n+1}, \cdots, \left( \Pi_{B_{y^1}} \tilde A^{-1}_{B_{y^1}} \tilde b(y) \right)_{n+p}\right)^\prime
\end{eqnarray*}
be two $p$-dimensional subvectors of the solution $\Pi_{B_{y^1}} \tilde A^{-1}_{B_{y^1}} \tilde b(y)$.
We can then directly compute the estimate of $\beta$ corresponding to this $y \in \mathcal{Y}^1$ by
$\beta(y) = X_{p}^{-1}\left( y_{p}-u_{p}(y)-v_{p}(y)\right)$.
Thus, we construct the subset estimate 
\begin{equation*}
\hat\Theta_0^\ast(\tau;y^1) = \left\{ X_{p}^{-1}\left( y_{p} - \left(\begin{array}{c}\left( \Pi_{B_{y^1}} \tilde A^{-1}_{B_{y^1}} \tilde b(y) \right)_1 \\\vdots\\ \left( \Pi_{B_{y^1}} \tilde A^{-1}_{B_{y^1}} \tilde b(y) \right)_p \end{array}\right) - \left(\begin{array}{c}\left( \Pi_{B_{y^1}} \tilde A^{-1}_{B_{y^1}} \tilde b(y) \right)_{n+1} \\\vdots\\ \left( \Pi_{B_{y^1}} \tilde A^{-1}_{B_{y^1}} \tilde b(y) \right)_{n+p} \end{array}\right) \right) : y \in \mathcal{Y}^1 \right\}.
\end{equation*}
Since this subset estimate is an image of $\mathcal{Y}^1$ through a simple linear transformation, it conveniently circumvents the need to solve the optimization problem for each $y \in \mathcal{Y}^1$.
Once this subset estimate has been computed, pick $y^2 \in \times_{i=1}^n [y_{Li},y_{Ui}] \backslash \mathcal{Y}^1$, use the simplex algorithm to get the set $B(y^2)$ of basic indices under $y^2$, and obtain the resultant subset estimate $\hat\Theta_0^\ast(\tau;y^2)$.
This is followed by the third step where $y^3 \in \times_{i=1}^n [y_{Li},y_{Ui}] \backslash \left( \cup_{k=1}^2 \mathcal{Y}^k \right)$ produces the subset estimate $\hat\Theta_0^\ast(\tau;y^3)$, and so on.
Repeat this process to obtain the set estimate $\hat\Theta_0^\ast(\tau) = \cup_{k=1}^K \hat\Theta_0^\ast(\tau;y^k)$ for $K$ steps until we exhaust $\times_{i=1}^n [y_{Li},y_{Ui}] = \prod_{k=1}^K \mathcal{Y}^k$.

\subsection{Geometric Properties of the Set of Best Linear Predictors}\label{sec:geometric_blp}

In this section, we provide some geometric properties of the set of best linear predictors proposed in Section \ref{sec:best_linear}.
It is shown that the set of connected, and therefore a projection of the set to each coordinate is interval-valued.
To show this property, we go through several auxiliary lemmas.
For the sake of rigorous proofs, we now formally define categorical sets and interval-valued sets below.

\begin{definition}\label{def:exclusive}
A random set $Y: \Omega \rightarrow \mathcal{K}(\mathbb{R})$ is categorical if
we have either
$Y(\omega) = Y(\omega^\prime)$
or
$Y(\omega) \cap Y(\omega^\prime) = \emptyset$
for all pairs $\omega,\omega^\prime \in \Omega$
\end{definition}

\begin{definition}\label{def:interval}
A random set $Y: \Omega \rightarrow \mathcal{K}(\mathbb{R})$ is interval-valued if
$Y(\omega)$ is an interval
for all $\omega \in \Omega$.
\end{definition}

For a random set $Y$, define the restricted set of selections
\begin{equation*}
Sel_C(Y) = \left\{ y \in Sel(Y) : F_y \text{ is continuous and strictly increasing on $F_y^{-1}((0,1))$}\right\}.
\end{equation*}
We first state the following auxiliary lemma of CDF equivalence for interval-valued categorical random sets.

\begin{lemma}\label{lemma:categorical}
Let $y_0, y_1 \in Sel_C(Y)$ be two selections from an interval-valued categorical random set (cf. Definitions \ref{def:exclusive} and \ref{def:interval}).
For any $\omega \in \Omega$, we have $F_{y_0}(Y(\omega)) = F_{y_1}(Y(\omega)).$
\end{lemma}

\begin{proof}
Define the set $\Omega_{L(\omega)} = \{\omega' \in \Omega : \sup Y(\omega') \leq \inf Y(\omega)\}$.
Let $\tau \in F_{y_0}(Y(\Omega))$.
By the strict increase of $F_{y_0}$ for $y_0 \in Sel_C(Y)$, we can write
$
\tau = Pr(\{\omega' \in \Omega : y_0(\omega') \leq F_{y_0}^{-1}(\tau)\})
$
where $F_{y_0}^{-1}(\tau) \in Y(\omega)$ by the definition of $\tau$.
Since $y_0 \in Sel_C(Y)$ and $Y$ is an interval-valued categorical random set, we the obtain
$
\tau = Pr(\{\omega' \in \Omega : y_0(\omega') \leq F_{y_0}^{-1}(\tau)\}) \geq Pr(\Omega_{L(\omega)})
$
by the monotone property of probability measures.

Assume by way of contradition that $\tau \not\in F_{y_1}(Y(\omega))$.
Since $F_{y_1}$ is increasing and $Y$ is interval-valued, this implies either $\tau < F_{y_1}(\zeta)$ for all $\zeta \in Y(\omega)$ or $\tau > F_{y_1}(\zeta)$ for all $\zeta \in Y(\omega)$.
Without loss of generality, we consider the former case, which can be rewritten as
$
\tau < Pr(\Omega_{y_1}^{\zeta})
$
for all $\zeta \in Y(\omega)$, where $\Omega_{y_1}^{\zeta} = \{\omega' \in \Omega : y_1(\omega') \leq \zeta\}$ for a short-hand notation.
Consider a decreasing sequence $\{\zeta_n\}_{n=1}^\infty \in Y(\omega)$ such that $\zeta_n \rightarrow \inf Y(\omega)$.
Since $y_1 \in Sel_C(Y)$, we have $\cap_{n=1}^\infty \Omega_{y_1}^{\zeta_n} = \Omega_{y_1}^{\inf Y(\omega)} \subset \Omega_{L(\omega)} \cup \{\omega' \in \Omega : y_1(\omega') = \inf Y(\omega)\}$
where the last set has a zero probability measure by the continuity of $y_1 \in Sel_C(Y)$.
Apply the continuity theorem of probability measures to the decreasing sequence $\{\Omega_{y_1}^{\zeta_n}\}_{n=1}^\infty$, we obtain
$
\tau < \lim_{n \rightarrow \infty} Pr(\Omega_{y_1}^{\zeta_n}) \leq Pr(\Omega_{L(\omega)}).
$
Combining this result with the conclusion from the last paragraph, we obtain
$
\tau \geq Pr(\Omega_{L(\omega)}) > \tau,
$
a contradiction.
Similarly, the case of $\tau > F_{y_1}(\zeta)$ for all $\zeta \in Y(\omega)$ leads to a contradiction.
Therefore, $\tau \in F_{y_1}(Y(\omega))$ holds.

A symmetric argument by interchanging the roles of $y_0$ and $y_1$ shows that $\tau \in F_{y_1}(Y(\omega))$ implies $\tau \in F_{y_0}(Y(\omega))$.
Therefore, $F_{y_0}(Y(\omega)) = F_{y_1}(Y(\omega))$ follows.
\end{proof}

\begin{lemma}\label{lemma:convex_combination}
Let $y_0, y_1 \in Sel_C(Y)$ be two selections from an interval-valued categorical random set $Y$ (cf. Definitions \ref{def:exclusive} and \ref{def:interval}).
Then, for any $\lambda \in [0,1]$, there exists $y_\lambda \in Sel_C(Y)$ such that $F_{y_\lambda} = (1-\lambda) \cdot F_{y_0} + \lambda \cdot F_{y_1}$.
\end{lemma}

\begin{proof}
Define $y_\lambda: \Omega \rightarrow \mathbb{R}$ by $y_\lambda(\omega) = \inf\left\{ y \in \mathbb{R} : F_{y_0}(y_0(\omega)) \leqslant (1-\lambda) \cdot F_{y_0}(y) + \lambda \cdot F_{y_1}(y) \right\}$.
First, we show that $y_\lambda$ is measurable.
$F = (1-\lambda) \cdot F_{y_0} + \lambda \cdot F_{y_1}$ is continuous and strictly increasing on its support because $y_0, y_1 \in Sel_C(Y)$.
Therefore, $F$ has a strictly increasing inverse $F^{-1}$ by \citet[pp. 266]{Pfeiffer90}, and it follows that $y_\lambda(\omega) = F^{-1} \circ F_{y_0} \circ y_0(\omega)$.
Since $F^{-1}$ and $F_{y_0}$ are continuous and $y_0$ is measurable, it follows that $y_\lambda$ is measurable.

Second, we show that $y_\lambda$ is a selection of $Y$. 
Let $\omega \in \Omega$.
Because $y_0 \in Sel_C(Y) \subset Sel(Y)$, we have $y_0(\omega) \in Y(\omega)$.
Thus, we obtain $F \circ y_\lambda(\omega) = F_{y_0} \circ y_0(\omega) \in F_{y_0} (Y(\omega)) = F_{y_1}(Y(\omega))$, where the first equality is due to the definition of $y_\lambda$ and the last equality is due to Lemma \ref{lemma:categorical}.
Taking a convex combination yields $F \circ y_\lambda(\omega) \in (1-\lambda) F_{y_0}(Y(\omega)) + \lambda F_{Y_1}(Y(\omega)) = F(Y(\omega))$.
Therefore, it follows that $y_\lambda(\omega) \in Y(\omega)$, showing that $y_\lambda$ is a selection of $Y$.

Finally, we show that $F_{y_\lambda} = (1-\lambda) \cdot F_{y_0} + \lambda \cdot F_{y_1}$.
This claim follows from the following chain of equalities:
$
F_{y_\lambda}(y)
=
P\left(\left\{\omega \in \Omega : y_\lambda(\omega) \leq y \right\}\right)
=
P\left(\left\{\omega \in \Omega : F^{-1} \circ F_{y_0}(y_0(\omega)) \leq y \right\}\right)
=
P\left(\left\{\omega \in \Omega : F_{y_0}(y_0(\omega)) \leq (1-\lambda) \cdot F_{y_0}(y) + \lambda \cdot F_{y_1}(y) \right\}\right)
=
(1-\lambda) \cdot F_{y_0}(y) + \lambda \cdot F_{y_1}(y),
$
where the first equality is by the definition of the cdf $F_{y_\lambda}$, the second equality is due to the definition of $y_\lambda$, the third equality is by the short-hand notation for $F = (1-\lambda) \cdot F_{y_0} + \lambda \cdot F_{y_1}$, and the last equality uses $F_{y_0}(y_0) \sim U(0,1)$ by the probability integral transform.
Therefore, we have $F_{y_\lambda} = (1-\lambda) \cdot F_{y_0} + \lambda \cdot F_{y_1}$.
\end{proof}

We now consider the joint random set $(x,Y): \Omega \rightarrow \mathbb{R}^p \times \mathcal{K}(\mathbb{R})$ and extend the restricted set $Sel_C(Y)$ of selections to 
\begin{eqnarray*}
Sel_C(x,Y) = \left\{ (x,y) \in Sel(x,Y) : \text{$F_{y \mid x}( \ \cdot \ \mid \xi)$ is continuous and strictly increasing} \right.
\\
\left. \text{ on $F_{y \mid x}^{-1}((0,1) \mid \xi)$ for each $\xi \in \text{Supp}(x)$}\right\}.
\end{eqnarray*}
Similar lines of a proof to those of Lemma \ref{lemma:convex_combination} yield the following extension to Lemma \ref{lemma:convex_combination}.

\begin{lemma}\label{lemma:convex_combination_x}
Suppose that $x$ has a countable support.
Let $(x,y_0), (x,y_1) \in Sel_C(x,Y)$ be two selections from an interval-valued categorical random set $Y$ (cf. Definitions \ref{def:exclusive} and \ref{def:interval}).
Then, for any $\lambda \in [0,1]$, there exists $(x,y_\lambda) \in Sel_C(x,Y)$ such that $F_{y_\lambda \mid x} = (1-\lambda) \cdot F_{y_0 \mid x} + \lambda \cdot F_{y_1 \mid x}$.
\end{lemma}

\begin{proof}
For each $\xi \in \text{Image}(x)$, let $\Omega_\xi = x^{-1}(\{\xi\}) \subset \Omega$.
Consider the restrictions $y_0 \mid_{\Omega_\xi}$, $y_1 \mid_{\Omega_\xi}$ and $Y \mid_{\Omega_\xi}$ of $y_0$, $y_1$ and $Y$, respectively, to the domain $\Omega_\xi$.
Because $\Omega_\xi = x^{-1}(\{\xi\})$ is a measurable subset of $\Omega$ due to the measurability of $x$, the restrictions $y_0 \mid_{\Omega_\xi}$ and $y_1 \mid_{\Omega_\xi}$ are also measurable functions.
Furthermore, $y_0 \mid_{\Omega_\xi}$ and $y_1 \mid_{\Omega_\xi}$ are selection of $Y \mid_{\Omega_\xi}$ because they are restricted to the identical domain $\Omega_\xi$.
Therefore, by Lemma \ref{lemma:convex_combination}, there exists a selection $y_{\xi,\lambda}: \Omega_\xi \rightarrow \mathbb{R}$ of $Y \mid_{\Omega_\xi}$ such that $F_{y_{\xi,\lambda}} = (1-\lambda) \cdot F_{y_0 \mid x=\xi} + \lambda \cdot F_{y_1 \mid x=\xi}$.

Define the function $y_\lambda: \Omega \rightarrow \mathbb{R}$ by the rule of assignment $y_\lambda(\omega) = y_{x(\omega),\lambda} (\omega)$.
Further, define the function $(x,y_\lambda): \Omega \rightarrow \mathbb{R}^{p+1}$ by the rule of assignment $(x,y_\lambda)(\omega) = (x(\omega), y_\lambda(\omega))$.
We have $(x,y_\lambda)(\omega) = (x(\omega),y_\lambda(\omega)) = (x(\omega),y_{x(\omega),\lambda}(\omega)) \in \{x(\omega)\} \times Y(\omega)$ because the previous paragraph concluded that $y_{\xi,\lambda}$ is a selection of $Y \mid_{\Omega_\xi}$ for each $\xi \in \text{Image}(x)$.
Because $y_{\xi,\lambda}$ is a measurable function for each $\xi \in \text{Image}(x)$ from the previous paragraph and $\text{Image}(x)$ is countable, this $y_\lambda$ is a measurable function.
But then, $(x,y_\lambda)$ is also a measurable function.
It also follows from the conclusion of the previous paragraph that $F_{y_\lambda \mid x=\xi} = F_{y_{\xi,\lambda}} = (1-\lambda) \cdot F_{y_0 \mid x=\xi} + \lambda \cdot F_{y_1 \mid x=\xi}$.
These arguments together show that the desired conclusion holds.
\end{proof}

The condition that the support of $x$ is countable is restrictive, but is needed in our proof. 
This condition guarantees that $y_\lambda$ is a measurable function.
Without this condition, it is not clear if the same conclusion remains due to the fact that a sigma field is closed only under countable unions.

Now, for a random vector $(x,y)$, we consider the best linear predictor $\beta_\tau$ defined by
\begin{equation*}
\beta_\tau = \arg\min_{\beta \in B} E[\rho_\tau(y - x^\prime \beta)]
\end{equation*}
where $\rho_\tau(u) = (\tau - 1[ u \leq 0]) \cdot u$ and $B \subset \mathbb{R}^p$ is a convex and compact set.

\begin{lemma}\label{lemma:convex}
Suppose that $x$ has a countable support.
Let $(x,y_0), (x,y_1) \in Sel_C(x,Y)$ be two selections from an interval-valued categorical random set $Y$ (cf. Definitions \ref{def:exclusive} and \ref{def:interval}).
If $E[\rho_\tau(y_0 - x^\prime \beta)]$ and $E[\rho_\tau(y_1 - x^\prime \beta)]$ are strictly convex in $\beta$, then, for any $\lambda \in [0,1]$, there exists $(x,y_\lambda) \in Sel_C(x,Y)$ such that $E[\rho_\tau(y_\lambda - x^\prime \beta)] = (1-\lambda) E[\rho_\tau(y_0 - x^\prime \beta)] + \lambda E[\rho_\tau(y_1 - x^\prime \beta)]$ is strictly convex.
\end{lemma}

\begin{proof}
Lemma \ref{lemma:convex_combination_x} shows that there exists $(x,y_\lambda) \in Sel_C(x,Y)$ such that $F_{y_\lambda \mid x} = (1-\lambda) \cdot F_{y_0 \mid x} + \lambda \cdot F_{y_1 \mid x}$.
Furthermore, $F_{y_\lambda \mid x} = (1-\lambda) \cdot F_{y_0 \mid x} + \lambda \cdot F_{y_1 \mid x}$ and the strict convexity of $E[\rho_\tau(y_0 - x^\prime \beta)]$ and $E[\rho_\tau(y_1 - x^\prime \beta)]$ in $\beta$ imply that $E[\rho_\tau(y_\lambda - x^\prime \beta)] = (1-\lambda) E[\rho_\tau(y_0 - x^\prime \beta)] + \lambda E[\rho_\tau(y_1 - x^\prime \beta)]$ is strictly convex.
\end{proof}

\begin{lemma}\label{lemma:smoothness}
Suppose that $x$ has a countable support.
If $(x,y) \in Sel(x,Y)$ is compactly supported and admits a conditional density $f_{y \mid x}( \ \cdot \ \mid \xi) \in C^2$ for each $\xi \in \text{Supp}(x)$, then
$E[\rho_\tau(y - x^\prime \beta)]$ is twice continuously differentiable in $\beta$ with
\begin{eqnarray*}
\frac{\partial}{\partial \beta} E[\rho_\tau(y - x^\prime \beta)] &=& E\left[ x \rho_\tau(y - x^\prime \beta) \frac{\partial \log f_{x,y}(x,y)}{\partial y}\right]
\quad\text{and}
\\
\frac{\partial^2}{\partial \beta \partial \beta^\prime} E[\rho_\tau(y - x^\prime \beta)] &=& E\left[ x \frac{\rho_\tau(y - x^\prime \beta)}{f_{x,y}(x,y)} \frac{\partial^2 f_{x,y}(x,y)}{\partial y^2} x^\prime \right]
\end{eqnarray*}
\end{lemma}

\begin{proof}
We first modify the check function $\rho_\tau$ by
$$
\bar\rho_\tau(u) =
\begin{cases}
\rho_\tau(u) & \text{if $u = \zeta - \xi^\prime \beta$ for some $(\xi,\zeta) \in \text{Supp}(x,y)$ and $\beta \in B$}
\\
0 & \text{otherwise}
\end{cases}
$$
With this modification, we have $\bar\rho_\tau \in L^1$ due to the compactness of $\text{Supp}(x,y)$ and $B$.
We can write the BLP objective as
\begin{eqnarray*}
E[\rho_\tau(y - x^\prime \beta)]
&=&
\sum_{\xi \in \text{Supp}(x)} \int \bar\rho_\tau(\zeta - q(\xi, \beta)) f_{x, y}(\xi, \zeta) d\zeta
\\
&=&
\sum_{\xi \in \text{Supp}(x)} \left(\bar\rho_\tau \ast f_{x,y}(\xi, \ \cdot \ ) \right) (q(\xi,\beta))
\end{eqnarray*}
where $q(x,\beta) = x^\prime \beta$ and `$\ast$' denotes the convolution operator.
Since $q$ is clearly twice continuously differentiable with respect to $\beta$ with its first and second derivatives given by $x$ and $0$, respectively, it suffices to show that $\left(\bar\rho_\tau \ast f_{x,y}(\xi, \ \cdot \ ) \right)$ is twice continuously differentiable with its first and second derivatives given by $ \bar\rho_\tau \ast \frac{\partial}{\partial y} f_{x,y}(\xi, \ \cdot \ ) $ and $ \bar\rho_\tau \ast \frac{\partial^2}{\partial y^2} f_{x,y}(\xi, \ \cdot \ ) $, respectively.
But this desired property follows from the fact that $f \in L^1$ and $g \in C^k$ implies $f \ast g \in C^k$ with $\partial^\alpha (f \ast g) = f \ast \partial^\alpha g$ for each $\alpha \in \{0,\cdots,k\}$, $\bar\rho_\tau \in L^1$, and our condition that $(x,y)$ admits $f_{x,y}(\xi, \ \cdot \ ) = f_{y \mid x}( \ \cdot \ \mid \xi) f_x(\xi) \in C^2$ for each $\xi \in \text{Supp}(x)$.
\end{proof}

We now define the set of best linear predictors by
$$
B_{I,\tau} = \left\{ \arg\min_{\beta \in B} E[\rho_\tau(y - x^\prime \beta)] : (x,y) \in Sel^\ast(x,Y)\right\}
$$
for
$$
Sel^\ast(x,Y) = \left\{ (x,y) \in Sel(x,Y) : \text{$(x,y)$ satisfies Condition \ref{condition:selection}} \right\},
$$
where the condition is given below.

\begin{condition}\label{condition:selection}
${}$\\
(i) $F_{y \mid x}( \ \cdot \ \mid \xi)$ is continuous and strictly increasing on $F_{y \mid x}^{-1}((0,1) \mid \xi)$ for each $\xi \in \text{Supp}(x)$.
\\
(ii) $(x,y)$ is compactly supported.
\\
(iii) $(x,y)$ admits a conditional density function $f_{y \mid x}( \ \cdot \ \mid \xi) \in C^2$ for each $\xi \in \text{Supp}(x)$.
\\
(iv)  $E[\rho_\tau(y - x^\prime \beta)]$ are strictly convex in $\beta$.
\end{condition}

\begin{proposition}
Suppose that $x$ has a countable support, and $Y$ is an interval-valued categorical random set $Y$ (cf. Definitions \ref{def:exclusive} and \ref{def:interval}).
If $E[\rho_\tau(y - x^\prime \beta)] = 0$ holds for some $\beta \in B$ for each selection $(x,y) \in Sel^\ast(x,Y)$,
then $B_{I,\tau}$ is connected.
In particular, the projection of $B_{I,\tau}$ to each coordinate is interval-valued.
\end{proposition}

\begin{proof}
Let $(x,y_0), (x,y_1) \in Sel^\ast(x,Y)$.
By Lemma \ref{lemma:convex_combination_x}, for any $\lambda \in [0,1]$, there exists $(x,y_\lambda) \in Sel_C(x,Y)$ such that $F_{y_\lambda \mid x} = (1-\lambda) \cdot F_{y_0 \mid x} + \lambda \cdot F_{y_1 \mid x}$.
Condition \ref{condition:selection} (ii) and (iii) are satisfied by such a selection $(x,y_\lambda)$ due to $F_{y_\lambda \mid x} = (1-\lambda) \cdot F_{y_0 \mid x} + \lambda \cdot F_{y_1 \mid x}$.
Furthermore, Lemma \ref{lemma:convex} shows that such a selection $(x,y_\lambda)$ also satisfies Condition \ref{condition:selection} (iv).
Therefore, $(x,y_\lambda) \in Sel^\ast(x,Y)$.

Let $U \supset B$ be an open subset of $\mathbb{R}^p$, $V=(0,1)$, and $W$ be an open subset of $\mathbb{R}^p$.
Define the function $\Psi: U \times (0,1) \rightarrow W$ by $\Psi(\beta,\lambda) = (1-\lambda) \frac{\partial}{\partial \beta} E[\rho_\tau(y_0-x^\prime\beta)] + \lambda \frac{\partial}{\partial \beta} E[\rho_\tau(y_1-x^\prime\beta)]$, which is guaranteed to exist by Lemma \ref{lemma:smoothness}.
Note also that $F_{y_\lambda \mid x} = (1-\lambda) \cdot F_{y_0 \mid x} + \lambda \cdot F_{y_1 \mid x}$, $(x,y_\lambda) \in Sel^\ast(x,Y)$ and Lemma \ref{lemma:smoothness} show that $\Psi(\beta,\lambda) = \frac{\partial}{\partial \beta} E[\rho_\tau(y_\lambda-x^\prime\beta)]$ for each $\lambda \in [0,1]$.

First, observe that for each $\lambda \in V$ there is exactly one $\beta_\tau(\lambda) \in U$ satisfying $\Psi(\beta_\tau(\lambda),\lambda) = \vec 0$ due to Lemma \ref{lemma:convex}.
Second, the local solvability (i.e., the existence of a continuous explicit function $\beta_\tau(\lambda)$ at each $\lambda \in V$) follows from the implicit function theorem with Lemmas \ref{lemma:convex} and \ref{lemma:smoothness}.
Third, for each compact subset of $K\subset V = (0,1)$, $\lambda \in K$ implies that $\Psi(\beta_\tau(\lambda),\lambda) = \vec 0$ holds for some $\beta_\tau(\lambda) \in B$ by the condition of the proposition.
Therefore, by Theorem 1 of \cite{Sandberg1981}, the map $\lambda \mapsto \beta_\tau(\lambda)$ from $V=(0,1)$ into $U$ is continuous.
Also, this continuity extends to the domain $[0,1]$ by the definition of $\Psi$ and Lemma \ref{lemma:smoothness}.

Therefore, $B_{I,\tau}$ is path-connected, and is therefore connected.
\end{proof}

\section{Estimation and Inference under Discrete Intervals}\label{sec:discrete}

This appendix section provides additional estimation and inference results that are relevant to a part of the procedure used in the empirical application in Section \ref{sec:empirical}.
Assume that $Y_{1},Y_{2},...$ are independently and identically distributed.
Let $Y_{i}=\left[ a_{i},b_{i}\right] $ and let $a_{\left( 1\right) }\leq a_{\left( 2\right) }\leq ...\leq a_{\left( n\right) }$ and $b_{\left( 1\right) }\leq b_{\left( 2\right) }\leq ...\leq b_{\left( n\right) }$ be the order statistics of $\left\{a_{i}\right\} _{i=1}^{n}$ and $\left\{ b_{i}\right\} _{i=1}^{n}$, respectively. 
For $t\in \mathbb{R}$, let $\lfloor t\rfloor $ denote the biggest integer smaller than $t$. Then, for $\tau \in \left( 0,1\right) $, define
\begin{eqnarray*}
\tilde{T}_{n}^{-1}\left( \tau \right) &=&a_{\left( \lfloor n\tau \rfloor \right) } 
\qquad\text{and}
\\
\tilde{C}_{n}^{-1}\left( \tau \right) &=&b_{\left( \lfloor n\tau \rfloor \right) }.
\end{eqnarray*}
to be our estimators of the lower bound and upper bound of $\Theta^Y_0(\tau)$, respectively.
Specifically, we define $\hat\Theta^Y_0(\tau) = [\tilde{T}_{n}^{-1}\left( \tau \right), \tilde{C}_{n}^{-1}\left( \tau \right)] = [a_{\left( \lfloor n\tau \rfloor \right) } , b_{\left( \lfloor n\tau \rfloor \right) }]$ as our estimator for $\Theta^Y_0(\tau)$.

In many surveys, the respondent is given a list of brackets to choose from. 
In this case both $a$ and $b$ are discrete random variables.
Assume $Y=\left[ a,b\right] $, and hence $\Theta^Y_0(\tau)=[q_{a}(\tau) ,q_{b}(\tau) ]$ by Theorems \ref{prop:id} and \ref{prop:id_convex}, where $q_{a}$ and $q_{b}$ vary discretely.
In the current appendix section, we define the set estimator by $\hat{\Theta}^Y(\tau)=[a_{\left( \lceil n\tau \rceil \right) },b_{\left( \lceil n\tau \rceil \right) }]$, where $\lceil t \rceil$ denotes the smallest integer greater than $t$. 

\subsection{Super-Consistency}\label{sec:superconsistency}

By Theorem 2 of \citet{RamachandramurtyRao1973}, we have
\begin{eqnarray*}
&& \mathbf{P}\left(r_n(a_{\left( \lceil n\tau \rceil \right) } - q_{a}(\tau)) \leq 1\right) \rightarrow 1,
\\
&& \mathbf{P}\left(r_n(a_{\left( \lceil n\tau \rceil \right) } - q_{a}(\tau)) \leq -1\right) \rightarrow 0,
\\
&& \mathbf{P}\left(r_n(b_{\left( \lceil n\tau \rceil \right) } - q_{b}(\tau)) \leq 1\right) \rightarrow 1,
\qquad\text{and}
\\
&& \mathbf{P}\left(r_n(b_{\left( \lceil n\tau \rceil \right) } - q_{b}(\tau)) \leq -1\right) \rightarrow 0,
\end{eqnarray*}
as $n \rightarrow \infty$, if $r_n \rightarrow \infty$ as $n \rightarrow \infty$.
Note that $r_n$ can diverge at an arbitrary rate as a function of $n$ -- even faster than $\sqrt{n}$.
Thus, we obtain the following result by the continuous mapping theorem.

\begin{theorem}\label{prop:discrete}
Suppose the random set takes the form $Y = [a,b]$ $\mathbf{P}$-a.s. where $a$ and $b$ are discretely distributed.
If $Y_1$, $Y_2,$ $\ldots$ are independently and identically distributed, then for $\tau \in (0,1)$
\begin{eqnarray*}
&&r_n H\left( \hat{\Theta}^Y(\tau),\Theta^Y_0(\tau)\right) \overset{P}{\rightarrow } 0
\qquad\text{and}
\\
&&r_n d_{H}\left( \hat{\Theta}^Y(\tau),\Theta^Y_0(\tau)\right) \overset{P}{\rightarrow 
} 0
\end{eqnarray*}
for $r_n \rightarrow \infty$ as $n \rightarrow \infty$, where $H$ and $d_H$ denote the Hausdorff distance and the directed Hausdorff distance, respectively.\footnote{
For two sets, $A$ and $B$, in a finite dimensional Euclidean space $\left(  \mathbb{R}^{k},\left\Vert {}\right\Vert \right) $, the directed Hausdorff distance from $A$ to $B$ is 
$
d_{H}\left( A,B\right) =\sup_{a\in A}\inf_{b\in B}\left\Vert a-b\right\Vert
$
and the Hausdorff distance between $A$ and $B$ is 
$
H\left( A,B\right) =\max \left\{ d_{H}\left( A,B\right) ,d_{H}\left(
B,A\right) \right\} .
$
}
\end{theorem}

Theorem \ref{prop:discrete} suggests that the estimator for the identification region when $Y$ is a discrete random set is super-consistent. 
Super-consistency is useful in cases where estimating a discrete quantile set is just a first step in a two-step estimation procedure. 
On the other hand, a drawback to this result is that we do not obtain a root-$n$ non-degenerate asymptotic normal distribution. 
The next subsection provides a modified estimator with a root-$n$ non-degenerate distribution.

\subsection{Non-Degenerate Asymptotic Distribution}

The lack of the ability to conduct inference with the naive quantile estimators is unfortunate.
However, in the special case where the discrete boundaries, $a$ and $b$, of the random set $Y$ are given as a count data, we can allow for inference even in the discrete case by using the idea of \citet{MachadoSilva2005}.
Suppose that $a$ and $b$ are supported in the set $\mathcal{J} = \{0,1,...,J-1\}$ of cardinality $J \in \mathbb{N}$.
Construct the random variables $\tilde a = a + u$ and $\tilde b = b + v$ where $u, v \sim \text{Uniform}(0,1)$ and $(u,v)$ is independent of $(a,b)$.
Let $F_{\tilde a,\tilde b}$ denote the joint cumulative distribution function of $(\tilde a,\tilde b)$, which is identified from the two-dimensional convolution of the distributions of $(a,b)$ and $(u,v)$.
As a result of the convolution, the distribution of $F_{\tilde a,\tilde b}$ is differentiable infinitely many times on $\left( \mathcal{J} \oplus (0,1) \right)^2$.
Furthermore, the above construction of the mixture distribution yields the marginal quantiles relations
\begin{align*}
q_{\tilde a}(\tau) &= q_a(\tau) + \frac{\tau - \sum_{j=0}^{q_a(\tau)-1} \text{Pr}(a=j)}{\text{Pr}(a=q_a(\tau))}
\\
q_{\tilde b}(\tau) &= q_b(\tau) + \frac{\tau - \sum_{j=0}^{q_b(\tau)-1} \text{Pr}(b=j)}{\text{Pr}(b=q_b(\tau))}
\end{align*}
See \citet{MachadoSilva2005}.
As such, we can define a new set estimator by
$$
\tilde \Theta^Y(\tau) = \left[\check a(\tau), \check b(\tau)\right] := \left[\tilde a_{\left( \lfloor n\tau \rfloor \right) } - \frac{\tau - \sum_{j=0}^{a_{\left( \lceil n\tau \rceil \right) }-1} \widehat{\text{Pr}}(a=j)}{\widehat{\text{Pr}}(a=a_{\left( \lceil n\tau \rceil \right) })}, \tilde b_{\left( \lfloor n\tau \rfloor \right) } - \frac{\tau - \sum_{j=0}^{b_{\left( \lceil n\tau \rceil \right) }-1} \widehat{\text{Pr}}(b=j)}{\widehat{\text{Pr}}(b=b_{\left( \lceil n\tau \rceil \right) })} \right]
$$
where $\widehat{\text{Pr}}(a=j)$ and $\widehat{\text{Pr}}(b=j)$ denote the empirical mass for each $j \in \mathcal{J}$.
Note that, in this estimator, we distinguish the $\sqrt{n}$-consistent estimator $(\tilde a_{\left( \lfloor n\tau \rfloor \right) }, \tilde b_{\left( \lfloor n\tau \rfloor \right) })'$ and the aforementioned super-consistent estimator $\sqrt{n}$-consistent estimator $( a_{\left( \lfloor n\tau \rfloor \right) },  b_{\left( \lfloor n\tau \rfloor \right) })'$ on purpose

To analyze the asymptotic distribution of this estimator $\tilde \Theta^Y(\tau)$, we first need the joint asymptotic distribution of the $2(J+1)$-dimensional vector $\sqrt{n}( \tilde a_{\left( \lfloor n\tau \rfloor \right) }-q_{\tilde a}(\tau), \tilde b_{\left( \lfloor n\tau \rfloor \right) }-q_{\tilde b}(\tau), \widehat{\text{Pr}}(a=0)-{\text{Pr}}(a=0),...,\widehat{\text{Pr}}(a=J-1)-{\text{Pr}}(a=J-1), \widehat{\text{Pr}}(b=0)-{\text{Pr}}(b=0),...,\widehat{\text{Pr}}(b=J-1)-{\text{Pr}}(b=J-1))'$ which consists stochastic element of the boundaries of the set estimator.
For any $\tau \in (0,1)$ such that $(q_{\tilde a}(\tau), q_{\tilde b}(\tau)) \in \left( \mathcal{J} \oplus (0,1) \right)^2$,
\begin{equation}
\sqrt{n}
\left(\begin{array}{c}
\left( \tilde a_{\left( \lfloor n\tau \rfloor \right) }-q_{\tilde a}(\tau), \tilde b_{\left( \lfloor n\tau \rfloor \right) }-q_{\tilde b}(\tau) \right)' 
\\
\left( \widehat{\text{Pr}}(a=0)-{\text{Pr}}(a=0),...,\widehat{\text{Pr}}(a=J-1)-{\text{Pr}}(a=J-1)\right)' 
\\
\left( \widehat{\text{Pr}}(b=0)-{\text{Pr}}(b=0),...,\widehat{\text{Pr}}(b=J-1)-{\text{Pr}}(b=J-1)\right)' 
\end{array}\right)
\overset{D}{\rightarrow}  N\left( \mathbf{0},\tilde \Sigma(\tau) \right),
\label{jointNormal_discrete}
\end{equation}
where $\tilde \Sigma(\tau)$ is a $2(J+1) \times 2(J+1)$ matrix which is completely expressed in (\ref{eq:sigma_discrete}) in Appendix \ref{sec:variance_matrix}.
If $\text{Pr}(a=q_a(\tau)) \neq 0$ and $\text{Pr}(b=q_b(\tau)) \neq 0$, then we therefore obtain
\begin{equation}\label{eq:z_check}
\sqrt{n} \left(\check a(\tau)-q_a(\tau), \check b(\tau) - q_b(\tau) \right)'
\overset{D}{\rightarrow}
\left(\check z_L(\tau), \check z_U(\tau) \right)'
\sim
N\left( \mathbf{0}, \Xi(\tau) \tilde \Sigma(\tau) \Xi(\tau)' \right),
\end{equation}
where $\Xi(\tau) = (\Xi_{1\cdot}',\Xi_{2\cdot}')'$ is a $2 \times 2(J+1)$ matrix, where the first row takes the form
\begin{equation*}
\Xi_{1\cdot} = \left( 1, \ \ 0, \ \ \left(\frac{1\{j \leq q_a(\tau)\}}{\text{Pr}(a=q_a(\tau))} + 1\{j=q_a(\tau)\} \frac{\tau - \sum_{j'=0}^{q_a(\tau)-1} {\text{Pr}}(a=j')}{{\text{Pr}}(a=q_a(\tau))^2}\right)_{j=0}^{J-1}, \ 0,...,0 \right)
\end{equation*}
and the second row takes the form
\begin{equation*}
\Xi_{2\cdot} = \left( 0, \ \ 1, \ \ 0,...,0, \ \  \left(\frac{1\{j \leq q_b(\tau)\}}{\text{Pr}(a=q_b(\tau))} + 1\{j=q_b(\tau)\} \frac{\tau - \sum_{j'=0}^{q_b(\tau)-1} {\text{Pr}}(b=j')}{{\text{Pr}}(b=q_b(\tau))^2}\right)_{j=0}^{J-1} \right)
\end{equation*}
From this asymptotic joint normal distribution, we obtain the following theorem that can be used for inference on random sets where the boundaries are counts.

\begin{theorem}
Suppose the random set takes the form $Y = [a,b]$ $\mathbf{P}$-a.s. where both $a$ and $b$ are discretely distributed with support contained in $\mathcal{J} = \{0,1,...,J-1\}$.
Construct the random variables $\tilde a = a + u$ and $\tilde b = b + v$ where $u, v \sim \text{Uniform}(0,1)$ and $(u,v)$ is independent of $(a,b)$.
For any $\tau \in (0,1)$ such that $(q_{\tilde a}(\tau), q_{\tilde b}(\tau)) \in \left( \mathcal{J} \oplus (0,1) \right)^2$, $\text{Pr}(a=q_a(\tau)) \neq 0$, and $\text{Pr}(b=q_b(\tau)) \neq 0$,
\begin{eqnarray}
&&
\sqrt{n}H\left( \check{\Theta}^Y(\tau),\Theta^Y_0(\tau)\right) \overset{D}{\rightarrow }%
\max \left\{ \left\vert \check z_{L}(\tau)\right\vert ,\left\vert \check z_{U}(\tau)\right\vert
\right\},  \label{Hasymptotics_discrete_check}
\\
&&
\sqrt{n}d_{H}\left( \check{\Theta}^Y(\tau),\Theta^Y_0(\tau)\right) \overset{D}{\rightarrow 
}\max \left\{ \left( \check z_{L}(\tau)\right) _{+},\left( \check z_{U}(\tau)\right) _{-}\right\},
\label{dHasymptotics_discrete_check}
\qquad\text{and}
\\
&&
n\left( d_{H}\left( \check{\Theta}^Y(\tau),\Theta^Y_0(\tau)\right) \right) ^{2}\overset{D}{%
\rightarrow }\max \left\{ \left( \check z_{L}(\tau)\right) _{+}^{2},\left( \check z_{U}(\tau)\right)
_{-}^{2}\right\},  \label{H2asymptotics_discrete_check}
\end{eqnarray}
where the random vector $\left( \check z_{L}(\tau),\check z_{U}(\tau)\right)$ is distributed according to (\ref{eq:z_check}).
\end{theorem}

\subsection{The variance matrix $\tilde\Sigma(\tau)$}\label{sec:variance_matrix}

In this section, we give a complete expression for the $2(J+1) \times 2(J+1)$ matrix, which is a component of the asymptotic normal distribution for the $2(J+1)$-dimensional vector $\sqrt{n}( \tilde a_{\left( \lfloor n\tau \rfloor \right) }-q_{\tilde a}(\tau), \tilde b_{\left( \lfloor n\tau \rfloor \right) }-q_{\tilde b}(\tau), \widehat{\text{Pr}}(a=0)-{\text{Pr}}(a=0),...,\widehat{\text{Pr}}(a=J-1)-{\text{Pr}}(a=J-1), \widehat{\text{Pr}}(b=0)-{\text{Pr}}(b=0),...,\widehat{\text{Pr}}(b=J-1)-{\text{Pr}}(b=J-1))'$.
See (\ref{jointNormal_discrete}).

$\tilde \Sigma(\tau)$ is a $2(J+1) \times 2(J+1)$ matrix
\begin{equation}\label{eq:sigma_discrete}
\tilde \Sigma(\tau) = 
\left(\begin{array}{cc}
\Sigma_{\tilde a,\tilde b}(\tau) & 
\begin{array}{cc}
\Sigma_{\tilde a, p_a}(\tau)' & \Sigma_{\tilde a, p_b}(\tau)'\\
\Sigma_{\tilde b, p_a}(\tau)' & \Sigma_{\tilde b, p_b}(\tau)'
\end{array}
\\
\begin{array}{cc}
\Sigma_{\tilde a, p_a}(\tau) & \Sigma_{\tilde b, p_a}(\tau)\\
\Sigma_{\tilde a, p_b}(\tau) & \Sigma_{\tilde b, p_b}(\tau)
\end{array}
&
\begin{array}{cc}
\Sigma_{p_a}(\tau) & \Sigma_{p_a, p_b}(\tau)\\
\Sigma_{p_a, p_b}(\tau)' & \Sigma_{p_b}(\tau)
\end{array}
\end{array}\right),
\end{equation}
$\Sigma_{\tilde a,\tilde b}(\tau)$ is a $2 \times 2$ matrix of the form
\begin{equation*}
\Sigma_{\tilde a,\tilde b}(\tau) =\left( 
\begin{array}{cc}
\frac{\tau \left( 1-\tau \right) }{f_{\tilde a}\left( q_{\tilde a}(\tau) \right) ^{2}} & \frac{F_{\tilde a,\tilde b}\left( q_{\tilde a}(\tau) ,q_{\tilde b}(\tau) \right) -\tau
^{2}}{f_{\tilde a}\left( q_{\tilde a}(\tau) \right) f_{\tilde b}\left(q_{\tilde b}(\tau) \right) } 
\\ 
\frac{F_{\tilde a, \tilde b}\left( q_{\tilde a}(\tau) ,q_{\tilde b}(\tau) \right) -\tau ^{2}}{f_{\tilde a}\left(q_{\tilde a}(\tau) \right) f_{\tilde b}\left( q_{\tilde b}(\tau) \right) } & \frac{\tau \left( 1-\tau \right) }{f_{\tilde b}\left(q_{\tilde b}(\tau) \right) ^{2}},
\end{array}
\right).
\end{equation*}
$\Sigma_{\tilde a, p_a}(\tau)$ and $\Sigma_{\tilde b, p_b}(\tau)$ is a $J \times 1$ are $J \times 1$ matrices of the forms
\begin{equation*}
\Sigma_{\tilde a, p_a}(\tau) = \left[ \frac{\left(1\left\{a \leq q_a(\tau)\right\} - \tau\right) \text{Pr}(a=j)}{f_{\tilde a}(q_{\tilde a}(\tau))} \right]_{j=0}^{J-1}
\text{ and }
\Sigma_{\tilde b, p_b}(\tau) = \left[ \frac{\left(1\left\{b \leq q_b(\tau)\right\} - \tau\right) \text{Pr}(b=j)}{f_{\tilde b}(q_{\tilde b}(\tau))} \right]_{j=0}^{J-1},
\end{equation*}
$\Sigma_{\tilde a, p_b}(\tau)$ and $\Sigma_{\tilde b, p_a}(\tau)$ are a $K \times 1$ matrices of the forms
\begin{align*}
\Sigma_{\tilde a, p_b}(\tau) &= \left[ \frac{\sum_{j'=0}^{q_a(\tau)} \text{Pr}(a=j',b=j) - \tau \text{Pr}(b=j)}{f_{\tilde a}(q_{\tilde a}(\tau))} \right]_{j=0}^{J-1}
\text{ and }\\
\Sigma_{\tilde b, p_a}(\tau) &= \left[ \frac{\sum_{j'=0}^{q_b(\tau)} \text{Pr}(a=j,b=j') - \tau \text{Pr}(a=j)}{f_{\tilde b}(q_{\tilde b}(\tau))} \right]_{j=0}^{J-1},
\end{align*}
$\Sigma_{p_a}$ and $\Sigma_{p_b}$ are $J \times J$ matrices of the forms
\begin{align*}
\Sigma_{p_a} &= \left[ 1\{j=j'\} \text{Pr}(a=j) - \text{Pr}(a=j)\text{Pr}(a=j') \right]_{j=0,j'=0}^{J-1, J-1}
\text{ and }\\
\Sigma_{p_a} &= \left[ 1\{j=j'\} \text{Pr}(b=j) - \text{Pr}(b=j)\text{Pr}(b=j') \right]_{j=0,j'=0}^{J-1, J-1},
\end{align*}
and $\Sigma_{p_a, p_b}$ is a $J \times J$ matrix of the form
\begin{align*}
\Sigma_{p_a, p_b} = \left[ \text{Pr}(a=j, b=j') - \text{Pr}(a=j)\text{Pr}(b=j') \right]_{j=0,j'=0}^{J-1, J-1}.
\end{align*}


\clearpage
\thispagestyle{empty}
\begin{figure}[htb]
	\centering
		(I) $\tau = 0.25$\\
		\includegraphics[width=0.60\textwidth]{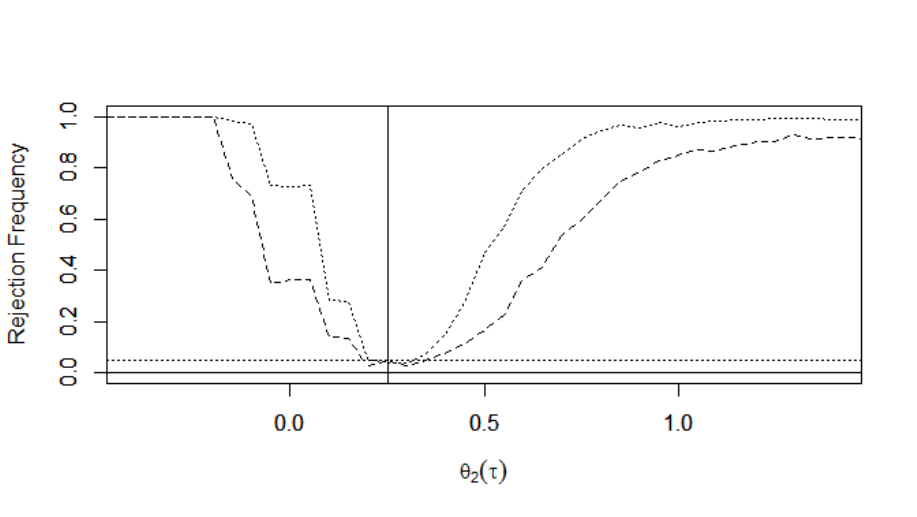}\\
		(II) $\tau = 0.50$\\
		\includegraphics[width=0.60\textwidth]{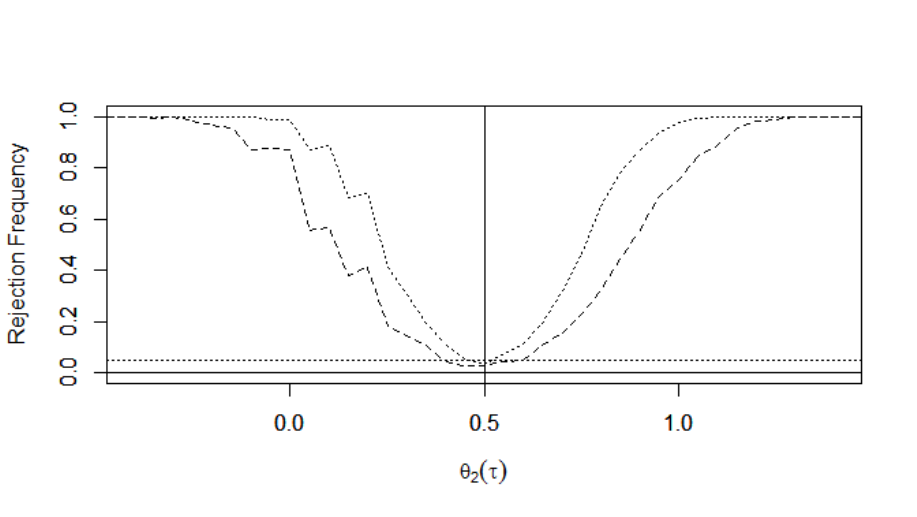}\\
		(III) $\tau = 0.75$\\
		\includegraphics[width=0.60\textwidth]{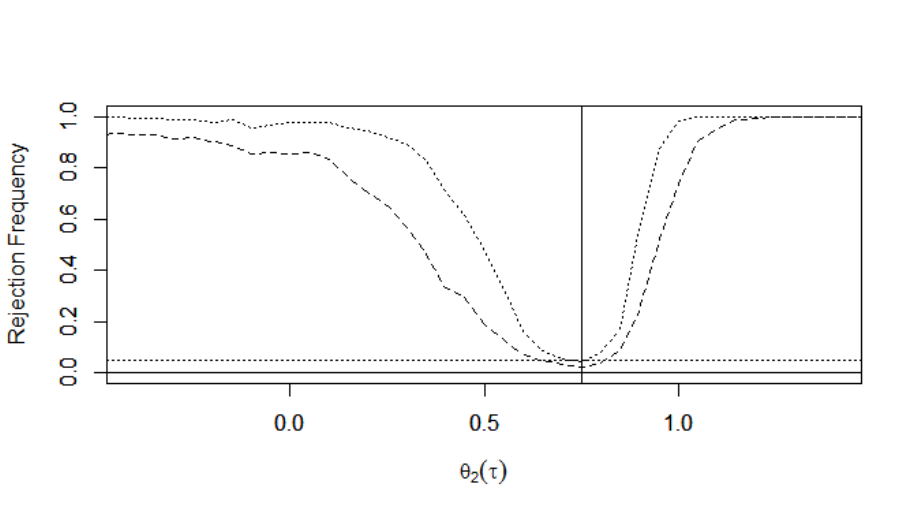}
	\caption{Rejection frequencies for inference of the parametric regression model (\ref{eq:mc_parametric}). The horizontal axis measures $\theta_2(\tau)$ given $\theta_1(\tau)$ fixed at (I) 1.25, (II) 1.50, and (III) 1.75. The dashed and dotted curves indicate the sample sizes of $n=100$ and 200, respectively.}
	\label{fig:mc_parametic}
\end{figure}

\newpage
\clearpage
\thispagestyle{empty}
\begin{table}
	\centering
		\begin{tabular}{rcccrrrr}
		\hline\hline
		         & (I) & (II) & (III) & \multicolumn{2}{c}{(IV)} & \multicolumn{2}{c}{(V)} \\
			$\tau$ & GT (2013) & $y=a$ & $y=b$ & \multicolumn{2}{c}{Set Estimate} & \multicolumn{2}{c}{95\% CR} \\
		\hline
		  0.10   & 0.244 (0.086) & -0.709 (0.196) & 1.260 (0.196) & [-10.780 & 11.331] & [-11.877 & 12.473]\\
			0.20   & 0.214 (0.085) & -0.506 (0.123) & 0.945 (0.123) & [-5.390  & 5.829]  & [-6.185  & 6.626]\\
			0.30   & 0.211 (0.084) & -0.201 (0.076) & 0.868 (0.076) & [-3.093  & 3.760]  & [-3.727  & 4.408]\\
			0.40   & 0.210 (0.083) & -0.019 (0.040) & 0.534 (0.040) & [-2.047  & 2.563]  & [-2.421  & 3.013]\\
			0.50   & 0.208 (0.083) & -0.013 (0.038) & 0.484 (0.038) & [-1.635  & 2.106]  & [-2.052  & 2.541]\\
			0.60   & 0.205 (0.084) &  0.046 (0.017) & 0.402 (0.017) & [-1.366  & 1.814]  & [-1.637  & 2.127]\\
			0.70   & 0.198 (0.085) &  0.062 (0.018) & 0.343 (0.018) & [-1.271  & 1.677]  & [-1.496  & 1.970]\\
			0.80   & 0.187 (0.081) &  0.057 (0.017) & 0.180 (0.017) & [-1.212  & 1.449]  & [-1.360  & 1.635]\\
			0.90   & 0.187 (0.077) &  0.054 (0.018) & 0.072 (0.018) & [-1.214  & 1.341]  & [-1.329  & 1.461]\\
		\hline\hline
		\end{tabular}
	\caption{Estimates, set estimates, and confidence regions for the coefficient of cleanup of hazardous waste sites on log house prices. Numbers in parentheses indicate estimated standard errors.}
		\label{tab:empirical_results}
\end{table}

\end{document}